\pdfoutput=1
\documentclass[12pt]{article}
\usepackage[utf8]{inputenc}

\usepackage{amsfonts, amsmath, amsthm, amssymb}
\usepackage{array, float, subcaption, xfrac}

\usepackage{xcolor,tikz, color}
\usetikzlibrary{fadings, patterns, shadows.blur, shapes, decorations.pathmorphing, positioning}
\usepackage{pgfplots}
\pgfplotsset{compat=1.12}
\usepackage[inline]{enumitem}

\usepackage{ifdraft} 
  

\DeclareMathOperator{\sign}{sign} 
\DeclareMathOperator*{\argmax}{arg\,max} 

\theoremstyle{plain}
\newtheorem{theorem}{Theorem}
\newtheorem{lemma}{Lemma}
\newtheorem{corollary}[lemma]{Corollary}
\newtheorem{proposition}[lemma]{Proposition}
\newtheorem{assumption}{Assumption}

\theoremstyle{remark}
\newtheorem{definition}{Definition}
\newtheorem{remark}{Remark}
\newtheorem{example}{Example}

\usepackage{hyperref}
\usepackage{natbib}
\usepackage[a4paper,
            margin=1.25in]{geometry}

\def\keywordname{{\bf Keywords:}}
\def\jelname{{\bf JEL:}}
\providecommand{\keywords}[1]{\def\and{{\textperiodcentered} }%
\par\addvspace\baselineskip
\noindent\keywordname\enspace\ignorespaces#1}%
\providecommand{\jel}[1]{\def\and{{\textperiodcentered} }%
\par
\noindent\jelname\enspace\ignorespaces#1}%

\title{Continuous Social Networks\thanks{We thank the WINE 2024 reviewers, Tristan Tomala, Frederic Koessler, and seminar participants for their comments. Juli\'an Chitiva gratefully acknowledges the support of Hi!PARIS and the ANR (``Investissements d'Avenir programme'') under grant ANR-18-EURE-0005/EUR DATA EFM. This work was partially done when Juli\'an Chitiva was visiting LUISS University in Rome. Xavier Venel's work was partially supported by the MIUR PRIN projects  ``Learning in Markets and Society'' (2022EKNE5K) and ``Supply Chain Disruptions, Financial losses and their prevention'' (PNRR P2022XT8C8). Xavier Venel is a member of Gnampa-Indam.}}
\author{Julián Chitiva\thanks{HEC Paris -- \href{mailto:julian.chitiva@hec.edu}{julian.chitiva@hec.edu}} and Xavier Venel\thanks{Luiss University -- \href{mailto:xvenel@luiss.it}{xvenel@luiss.it}}}
\date{}

\begin{document}
\maketitle
\begin{abstract}
    We develop an extension of the classical model of \cite{degroot1974} to a continuum of agents when they interact among them according to a DiKernel. 
    We show that, under some regularity assumptions, the continuous model is the limit case of the discrete one. 
    Additionally, we establish sufficient conditions for the emergence of consensus.
    We provide some applications of these results. 
    First, we establish a canonical way to reduce the dimensionality of matrices by comparing matrices of different dimensions in the space of DiKernels.
    Then, we develop a model of Lobby Competition where two lobbies compete to bias the opinion of a continuum of agents.
    We give sufficient conditions for the existence of a Nash Equilibrium and study their relation with the equilibria of discretizations of the game.
    Finally, we characterize the equilibrium for a particular case of DiKernels.  
\end{abstract}
\jel{D85 \and C72}
\section{Introduction}

The dynamics of opinions play a crucial role in shaping economic outcomes, decision-making processes, and policy choices.
Recent political events, such as Brexit or the 2016 U.S. presidential election, have surprised political scientists and economists.
The increasing popularity of social network platforms like Facebook  contributed significantly to these outcomes.
For example, Facebook has shown a sustained growth in Daily Active Users (DAUs) jumping from 483 million in 2011 to 2.06 billion in 2023 \citep{meta2012,meta2023}.
Economic Theory needs a comprehensive theoretical framework to understand the complexities of opinion dynamics in large economic systems.
This paper synthesizes the theoretical frameworks of the DeGroot model, the models of lobby competition, and DiKernel Theory, into an integrated approach that captures comprehensively the dynamics of opinion formation under lobby influence. 
Each of these frameworks offers valuable insights into distinct aspects of opinion formation, and their integration promises a more holistic understanding of the complex dynamics that shape public opinion.
Furthermore, integrating DiKernel theory offers valuable mathematical tools to analyze the complex structure of large economic networks.
DiKernels allow us to characterize and understand the relationships between economic actors in large societies, which, in turn, impact the spread of opinions, the formation of economic consensus, and the emergence of polarization.

Originally developed in \cite{degroot1974} for social psychology, the DeGroot model provides a fundamental approach to studying opinion dynamics within social networks.
In this model, the social structure is described by a weighted directed network in which agents hold opinions about some topic of interest (e.g. politics, brands, etc.).
It assumes a discrete-time updating process in which individuals' new opinion is the average of their neighbors' opinions.
While the DeGroot model offers a solid foundation for social learning, it faces computational restrictions when taking large networks into consideration, which are important for capturing real-world opinion dynamics.
To address these limitations, we develop a novel extension of the model of DeGroot to networks with a continuum of agents.
The ``continuous DeGroot model" constitutes the main contribution of this paper and our first result. This model captures the dynamics of opinion evolution in the continuous case. 
A natural question is whether consensus arises. 
\cite{golub2010naive} show that under certain conditions on the adjacency matrix, this updating process leads to a consensus in the opinions over time.
To study this question for the continuous setup, we associate the DiKernel with a Markov chain on a continuous space state. 
By analyzing the ergodic properties of these Markov chains, we establish sufficient conditions for consensus emergence, constituting our second result.
Additionally, we study the relationship between the continuous and discrete DeGroot models. 
We begin by drawing the parallel between matrices, that represent weighted directed networks, and their extension to a continuum of agents, called DiKernels. 
Then, we show that under Lipschitz\footnote{The Lipschitz condition is imposing some regularity to the function that rule social interactions.} conditions on the DiKernel, the continuous DeGroot model is the limit case of the discrete one as the number of agents goes to infinity. This constitutes our third result.  

This work contributes to the growing strand of literature in economics that uses Graphons (and in general DiKernels) to study the interactions between non-atomic agents in large economic environments (see \cite{erol2023contagion,aurell2022finite,shah2019reducing,borgs2017graphons}) and characterizes the equilibrium sets of these games (see \cite{pariseOzdaglar_graphonGames,rokade2023graphon}).
Moreover, this extension allows us to study economic scenarios that involve interactions between a large number of agents through a continuous approximation. 
We propose some applications of our model both in non-strategic and strategic setups. 
On the non-strategic side, we propose a method to reduce the dimensionality of discrete DeGroot problems by quantifying the error one is making when  considering aggregation of agents in groups.
This process of dimensionality reduction constitute our fourth result.
Then, we propose an application to strategic setups, which drives the economic implications of our model. 
We consider the problem of competition between lobbies that aim at influencing economic decision-making in a society, through strategic campaigns and resource mobilization.  
This framework acknowledges that opinions are not solely determined by information exchange but also by persuasion, cognitive biases, and social influence. 
Moreover, it provides a more realistic representation of opinion formation in economic and political contexts. 
Understanding the mechanics of lobby competition is essential to comprehend how interest groups shape public opinion and affect economic policies. 
There have been different approaches to model this competition, particularly in the discrete case.
For example, \cite{grabisch2018strategic}
model the competition of lobbies with opposing interests by modifying the structure of the graph that represents the society in a static fashion. 
\cite{mandel_venel} provide a more dynamic approach to this problem. 
We propose a one-shot game between two opposing lobbies who want to bias the opinion of a continuum of \emph{non-strategic} agents in their favor.
In this game, each lobby has a budget constraint, and chooses a function that determines the level of influence the lobby is exerting on each agent.
Then, based on the influence of each lobby, agents change their initial opinion into a new one before updating it \textit{\`a la} DeGroot.
As in \cite{venel2021regularity}, the lobbies’ utility function is a discounted sum of the opinions of the agents, following the intuition that lobbies would prefer that a large mass of agents get to their ideal opinion as soon as possible.
Our fifth result shows the existence of the Nash Equilibrium under some sufficient conditions on the strategy Sets and the ``Competition Operator''.
Finally, our sixth result studies conditions under which a Nash Equilibrium of the game is an $\varepsilon$\nobreakdash-Nash Equilibrium of the discretization of the game.

The remainder of the paper is organized as follows. 
In Section~\ref{sec:continuous_social_dynamics}, we establish the connection between the discrete and continuous versions of the DeGroot model. 
In particular, we give sufficient conditions for the emergence of consensus, we study the relationship between the discrete and continuous DeGroot models, and we study the problem of dimensionality reduction. 
In Section~\ref{sec:competition}, we provide a game-theoretical model of lobby competition to influence a continuum of agents, prove the existence of the Nash Equilibria of the game, and give the characterization of the equilibrium strategies for particular setups of the game. 
Moreover, we give conditions under which a Nash Equilibrium of the continuous case are also an $\varepsilon$\nobreakdash-Nash Equilibrium of the discretization of the game.  
Section~\ref{sec:conclusion} concludes. 
All proofs are in the appendix.

\subsection{Related Literature}

The competition between lobbies to influence a set of agents is of great economic interest and has been widely studied. 
First, economists reduced the question to the problem of one lobby that wants to influence a set of agents and needs to identify the key player. 
This notion of key-player was first studied by \cite{bonacich1972factoring} and \cite{freeman1977set}, and more recently in works as \cite{ballester2006s} or \cite{banerjee2013diffusion}.
Some natural directions for extending this one\nobreakdash-player analysis is, to increase the number of lobbies, or how lobbies influence the agents across time. 
These questions have been approached in several articles when the set of agents is discrete. 
For example, \cite{bimpikis2016competitive} study two firms who allocate their marketing budget at the beginning of the game, to influence a set of agents that exchange some piece of information in each period.
\cite{GOYAL201958} choose a contagion framework,  where each lobby chooses a seed from which their opinion propagates in the network.
The closest models to our own are those studied by \cite{lever2010strategic}, \cite{grabisch2018strategic}, \cite{mandel_venel}, and \cite{venel2021regularity}. 
They are based on the model of \cite{degroot1974}, which proposes a na\"ive (non-Bayesian) information transmission model between players and represents a non-strategic way of modeling dynamics of opinions inside a network. 
Opinions are represented by numbers between 0 and 1, and agents update them in discrete time as a weighted average of the opinions at the previous stage of their neighbors. 
\cite{grabisch2018strategic} models the competition of lobbies with opposing interests by modifying the structure of the graph that represents the society. 
Each lobby forms a link with a member of society, in order to change the dynamics of opinions that evolve \emph{\`a la} DeGroot according to this new network structure. 
In this model, the lobbies only act at the beginning of the game in a static manner. 
In \cite{mandel_venel}, the authors depart from the static influence and study the dynamic competition between lobbies who change their target during the game.
Although the lobbies' influence affects the weights that the target agent uses in his updating process, the network structure does not change, as in \cite{grabisch2018strategic}.
\cite{venel2021regularity} proposes a class of stochastic games called dynamic opinion games. These games propose a general model in which lobbies can exert unrestricted influence on the network, relaxing the assumption that lobbies perturb an existing network as in \cite{grabisch2018strategic}.
 
So far, the literature has modeled influence competition for a discrete set of agents.
Nevertheless, social networks have increased exponentially in size, posing a computational problem for analyzing network models.
Therefore, we need to model social networks when the number of agents diverges. 
\cite{lovaszSzegedy_limits} gave the mathematical foundations of Graphon Theory, which studies the family of bounded symmetric measurable functions $W:[0,1]^2\to [0,1]$, known as Graphons. 
These define the limit of a sequence of graphs when the number of nodes tends to infinity. 
Extensions of these objects, were studied by \cite{lovasz2012large}, \cite{borgs2008convergent,borgs2012convergent}, and \cite{lovasz2013non}, when the codomain is extended (Kernels) or when the symmetry condition is dropped (DiGraphons and DiKernels). 
Furthermore, \cite{avella2018centrality} explored the connection between discrete graphs and graphons by studying measures of centrality of these continuous objects. 

Graphons have provided different tools to model interactions between non-atomic agents.
\cite{caines2019graphon} suggest that graphons are an extension of mean-field games to heterogeneous settings. 
In the works by \cite{pariseOzdaglar_graphonGames} and \cite{rokade2023graphon}, the authors consider a class of games in which a continuum of agents interacts according to a Graphon in a static manner. 
Moreover, they provide conditions for the existence of equilibrium strategies and convergence between these strategies and their discretization. 
Concerning multi-stage analysis, \cite{erol2023contagion} use  Graphons as a statistical model to generate random graph to study threshold contagion dynamics and the associated optimal seeding problems. 
In this case, the opinions are discrete different from our approach.
The Graphon structure has been successfully applied in other areas of economic interest such as epidemics as in \cite{aurell2022finite}, crowdsourcing as in \cite{shah2019reducing}, estimation of social and information networks as in \cite{borgs2017graphons} and community detection as in \cite{eldridge2016graphons} to study the interaction among large populations.

Dynamics of opinions in Graphons has been previously studied in \cite{aletti2022opinion}, which makes it the closest to our model.
The authors focus only on piecewise constant graphons and model the opinions of agents using mean-field theory. 
Our model differs from theirs on several dimensions. 
First, we extend the dynamics of opinions to a more strategic setting that has economic applications.
Second, we consider DiKernels instead of Graphons, which allows analyzing more general setups, including those in which agents do not behave in a reciprocal and symmetric manner. 
Additionally, our model considers a broader class of DiKernels beyond the piecewise constant case, though this specific case helps link continuous and discrete setups.

Finally, dimensionality reduction has been studied via stochastic network modeling. In \cite{golub2012homophily,golub2012homophily2} the authors propose the Multi-type Random Networks model where society is represented by a symmetric adjacency matrix. The authors achieve the reduction of dimensionality by studying the matrix of expected fractions of links between the different groups. This matrix simplifies the analysis by working with groups instead of individuals. In this paper we propose a similar approach, since Graphons could be seen as a stochastic network formation model, as established in \cite{pariseOzdaglar_graphonGames}. Nevertheless, our work extends this analysis to non-symmetric graphs and proposes a canonical way to compare matrices of different dimensions.

\subsection{Notation} 
\begin{enumerate*}[label=(\roman*)]
\item We use lowercase letters to denote scalars and scalar-valued functions (e.g. $f$ and $f(x)$). 
\item We use lowercase letters with a hat ( $\hat{}$ ) to denote vectors (e.g. $\hat{f}$).
\item We denote uppercase letters with a hat ( $\hat{}$ ) to denote matrices (e.g. $\widehat{W}$).
\item We denote by $[n]$ the set $\{1,2,\dots,n\}$.
\end{enumerate*}

\section{Continuous Social Dynamics}\label{sec:continuous_social_dynamics}

In this section, we introduce the continuous DeGroot dynamic and analyze its relation to the discrete version. We provide a canonical way to transition between discrete and continuous versions of the model. Finally, we state a theorem of convergence between these two dynamics.

\subsection{Model}

Consider the classical model of \cite{degroot1974} in which a set $[n]$ of \emph{non-strategic} agents interact in a social network. Each agent $i\in [n]$ holds an initial opinion on a subject represented by a number ${f_0}(i)\in[-1,1]$. 
The social network is represented by an $n\times n$ non-negative matrix $\widehat{W}=[w_{ij}]$, such that for all $i\in[n]$, $\sum_{j=1}^n w_{ij}=1$. 
Such matrices will be called row-stochastic. 
The opinions are updated at discrete steps according to $\hat{f}_t = \widehat{W} \hat{f}_{t-1}={\widehat{W}}^t \hat{f}_0$, where $\hat{f}_t={({f_{t}}_i)}_{i\in [n]}$ is the opinion vector at time $t$. 
We can interpret each coefficient $w_{ij}$ as the weight that agent $i$ places on the opinion of agent $j$, and the updating process as agents taking a weighted average of their neighbors' opinions in the previous period. 

Now consider a continuum of \emph{non-strategic} agents indexed by $x \in [0,1]$ that interact in a social network. Each agent holds an initial opinion $f_0(x)\in[-1,1]$ and updates it at discrete time following a continuous version of the classical DeGroot model.\footnote{Without loss of generality the model can be extended to consider opinions in $[D,D^\prime]$ by adapting all subsequent definitions.}
To establish the continuous dynamic, we need a structure to determine the interactions between the agents and a rule to update their opinions. 

To model the interaction, we consider a bounded measurable function $W:[0,1]^2 \to [0,M]$, where $M>1$. This kind of functions is known as DiKernel, and they extend directed graphs as the number of agents diverges ($n \to \infty$). 
Here, $W(x,y)$ represents the weight that agent $x$ places on agent $y$'s opinion. To establish the parallel between the discrete and continuous versions, we focus on DiKernels $W$ such that for every $x$ it holds $\int_0^1 W(x,y)\mathrm{d}y=1$. 
Such DiKernels will be called row-stochastic DiKernels. Throughout this paper, we will assume all DiKernels to be row-stochastic.
The mathematical foundations of these objects started with \cite{lovaszSzegedy_limits} who established the Theory of Graphons, which are bounded symmetric measurable functions $W:[0,1]^2\mapsto [0,1]$. The works of \cite{lovasz2012large}, \cite{borgs2008convergent,borgs2012convergent} and \cite{lovasz2013non} study extensions of these objects where the codomain is $\mathbb{R}$ (known as Kernels) and when the symmetry conditions are dropped (known as DiGraphon and DiKernels). 

As in the classical DeGroot model, agents need a rule to update their function of opinions  $f:[0,1]\mapsto [-1,1]$. For this purpose, we define the linear operator $\operatorname{T}(W)$ for a row-stochastic DiKernel $W$ as the integral operator $\operatorname{T}(W):L^1\left(\left[0,1\right]\right) \rightarrow L^1\left(\left[0,1\right]\right)$ as $(\operatorname{T}(W)f)(x) = \int_0^1 W(x,y) f(y) \mathrm{d}y$.
The space $L^1\left(\left[0,1\right]\right)$ is the set of functions whose absolute values have a finite integral over $[0,1]$.
This operator averages out the opinions of the agents according to the weights given by $W$. 
It is worth noticing that, $(\operatorname{T}(W)f)(x)$ is a function of opinions resulting from the weighted average of a function of opinions $f(x)$ using the DiKernel $W$ as weights.
Since this is a new function of opinions we can apply the operator $\operatorname{T}(W)$ to it, and it would give us the composition of the operators, hence we can state the recursive property as in Equation~(\ref{eq:degroot_cont_recursive}).

\begin{definition}[Continuous DeGroot Model]\label{def:degroot_cont}
Consider a continuum of \emph{non-strategic} agents indexed by $x \in [0,1]$. Let $W$ be the row-stochastic DiKernel that rule the social interactions between them.
Each agent has an initial opinion on a subject represented by a number $f_0(x)\in[-1,1]$ and updates his opinion at discrete time following Equation~(\ref{eq:degroot_cont_recursive}). We will refer to the sequence $(f_t)_{t\in\mathbb{N}}$ as the dynamics of opinions under $W$ from $f_0$.
\begin{equation}
    f_{t}(x)=\left(\operatorname{T}(W)f_{t-1}\right)(x)=\left(\operatorname{T}^t(W)f_0\right)(x) \label{eq:degroot_cont_recursive}
\end{equation}
\end{definition}

\begin{remark}
An equivalent approach is by defining the power DiKernel $W^t$. To do this, we define the product between DiKernels $W$ and $V$ as $(W\ast V)(x,y) = \int_0^1 W(x,u) V(u,y) \mathrm{d}u$, which can be interpreted as an extension to the continuous case of matrix multiplication.
These two approaches are conceptually different, since the second changes how agents interact with each other. 
Nevertheless, they can be shown to be equivalent using the product between DiKernels and Fubini's theorem to prove that $\operatorname{T}^t(W)=\operatorname{T}(W^t)$.
\end{remark}

\subsection{Sufficient condition for emergence of consensus}\label{sec:conditions_consensus}

One of the key questions in the literature on propagation of opinions in networks is the question of emergence of consensus. This question has been studied extensively in the mathematical literature and applied to economics in \cite{golub2010naive} for the discrete case. We now provide a sufficient condition in the context of the continuous DeGroot model by associating the DiKernel to a Markov Chain on a continuous state space. The emergence of consensus is connected with the existence of the stationary distribution. This condition offers only a partial characterization.  Determining whether these Markov chains have a stationary distribution, and whether they converge to it, remains challenging.

\begin{assumption}\label{assump:toto_mixing}
We say that the DiKernel $W$ is $\gamma$-mixing if for all $x,y\in[0,1]$, $W(x,y)\geq \gamma$.
\end{assumption}

\begin{proposition}\label{prop:existence_unitype}
Let $\gamma>0$ and $W$ be a DiKernel that satisfies Assumption~\ref{assump:toto_mixing}. 
Then, there exists $h:[0,1] \rightarrow \mathbb{R}^*_+.$
\begin{itemize}
    \item $h$ is the density of a probability distribution over $[0,1]$,
    \item the opinions $f_t$ under $W$ starting from $f_0$ converges to the constant function $f^*$ defined by
    \[
    f^*=\int\limits_0^1 h(y)f_0(y)dy,
    \]
    \item there exist $\rho$ and $\alpha$ such that
    \[
    \sup_{x \in [0,1]} |f_t(x)-f^*| \leq \alpha \rho^t.
    \]
    \end{itemize}
\end{proposition}

\noindent One can associate the DiKernel $W$ to a Markov chain defined as follows:
\[
\forall y\in X, \text{ and } A \text{ Borel set of $[0,1]$},\ K(y,A)=\int_{x\in A} W(x,y)dx
\]
and given an initial distribution $p_0$ over $[0,1]$, the sequence of probability distribution $(p_t)_{t\geq 0}$ generated by $K$. The condition that $W$ is $\gamma$-mixing implies that the Markov chain $K$ admits a Doeblin minorization:
\[
\forall x\in [0,1],\ \int_{y \in A} w(x,y)dy\geq \gamma \lambda(A),
\]
where $\lambda$ is the Lebesgue measure. This type of condition was investigated by Doeblin, Kolmogorov, and Doob in the early 20th century and synthesized in Doob's \cite{Doob_1953}. It implies that the sequence $(p_t)_{t\geq 1}$ convergence geometrically toward a unique invariant distribution for the Markov Chain $K$ (see Theorem 5.4 \cite{Cappe_2005} for a modern proof). Denote by $h$ the limit, it implies immediately the emergence of consensus and the geometric rate of convergence of the opinions.

\subsection{From Discrete to Continuous}\label{sec:disc_cont}

There is a natural way to express any discrete DeGroot model as a continuous one by using a ``block-constant'' DiKernel. This concept involves splitting the continuum of agents into a discrete number of communities where there is a fixed weight for agents inside each community and a possibly different weight for agents in the other communities. Before transitioning from the discrete to the continuous framework, it is important to clarify that this transition is not as straightforward as one might assume for two reasons. First, there is no \emph{a priori} uniqueness of the associated block-constant DiKernel. Second, discrete and continuous objects represent different concepts. In the discrete DeGroot model, row-stochastic matrices represent probabilities (i.e., the weights in the DeGroot model). On the other hand, in the continuous model, due to the non-atomic nature of agents, DiKernels represent densities with respect to the uniform measure. 

The key to connecting these two concepts is to fix a partition of $[0,1]$ with size equal to the dimension of the DeGroot matrix. 
Naturally, the Lebesgue measures of the elements of the partition induce a vector of weights $\hat{p}$. 
We use this vector of weights to rescale the values of the probability matrix from the discrete DeGroot model and transform it into one that represents densities with respect to the measure induced by $\hat{p}$, which will represent our DiKernel. 
This procedure is described formally in Definition~\ref{def:discrete_cont_Dikernel} and exemplified in Examples~\ref{ex:disc_cont_dikernel} and~\ref{ex:disc_cont_weight_degroot}.
Using this equivalence, we show in Proposition~\ref{prop:discrete_opinion_dynamics} that there exists a block-constant DiKernel in the continuous framework that mimics the discrete DeGroot dynamics. 
The proof of this result can be derived directly from the construction of these objects.
Moreover, the average opinion in the continuous model is related to the weighted average of the discrete model. 
We exemplify this result in Examples~\ref{ex:disc_cont_degroot_dynamics} and~\ref{ex:disc_cont_weighted_degroot_dynamics}.

\begin{definition}[Discrete to Continuous]\label{def:discrete_cont_Dikernel}
    Let $\widehat{W}=[w_{ij}]$ be a $n\times n$ row-stochastic matrix, $\hat{f_0}$ be an initial vector of opinion and $\mathcal{V}=(V^i)_{i \in [n]}$ be an interval partition of $[0,1]$. We define for every $i \in [n]$, $\hat{p}_i=\lambda(V^i)$ and the block-constant DiKernel $W_{\mathcal{V}}$ by
    \[
    W_{\mathcal{V}}(x,y)= \frac{w_{ij}}{\hat{p}_j}=w^\prime_{ij} \text{ for } (x,y)\in V^i\times V^j
    \]
    and the initial function of opinions
    \[ {f_\mathcal{V}}_0(x)=\hat{f_0}_i\text{ for } x\in V^i. 
    \]
\end{definition}

\begin{proposition}\label{prop:discrete_opinion_dynamics} 
 Let $\widehat{W}$ be a row-stochastic matrix, $\hat{f_0}$ be an initial vector of opinion and $\mathcal{V}=\{\mathcal{V}^i\}_{i \in I}$ be an interval partition of $[0,1]$.
\begin{itemize}
    \item The (continuous) dynamic of opinions under $W_{\mathcal{V}}$ from $f_{\mathcal{V},0}$ coincides with the (discrete) dynamic of opinions under $\widehat{W}$ from $\hat{f_0}$.
    \item The average opinion at stage $t$ in the continuous DeGroot is the $\hat{p}$-weighted average opinion in the discrete DeGroot.
\end{itemize}
\end{proposition}

\begin{example}[Discrete to Continuous: Uniform Partition]\label{ex:disc_cont_dikernel}
Consider the matrix $\widehat{W}=\left(\begin{smallmatrix}0 && 1/2 && 1/2 \\ 1 && 0 && 0 \\ 0 && 1 && 0 \\ \end{smallmatrix}\right)$ and the uniform partition $\mathcal{P}_3=\{[0,\sfrac{1}{3}), [\sfrac{1}{3},\sfrac{2}{3}),[\sfrac{2}{3},1]\}$ of $[0,1]$.
This partition induces the vector of weights $\hat{p}=\left(\sfrac{1}{3},\sfrac{1}{3},\sfrac{1}{3}\right)$. 
Since $\widehat{W}$ represents probabilities we need to transform it in one that represents densities with respect to the measure induced by $\hat{p}$. 
Since the weights are uniform, we define the corresponding density matrix $\widehat{W}^\prime=[w^\prime_{ij}]$ by dividing the matrix $\widehat{W}$ by $\frac{1}{3}$. Therefore, $\widehat{W}^\prime=3\widehat{W}=\left(\begin{smallmatrix}0 && 3/2 && 3/2 \\ 3 && 0 && 0 \\ 0 && 3 && 0 \\ \end{smallmatrix}\right)$. Finally, define the DiKernel $W$ as $W_{\mathcal{P}_3}(x,y)=w^\prime_{ij}$ for $(x,y)\in {\mathcal{P}_3}^i\times {\mathcal{P}_3}^j$ with graphical representation is shown in Figure~\ref{fig:disc_cont_weight_degroot_P}.
\end{example}

\begin{example}[Discrete to Continuous: General Partition]\label{ex:disc_cont_weight_degroot}
Consider the matrix $\widehat{W}=\left(\begin{smallmatrix}0 && 1/2 && 1/2 \\ 1 && 0 && 0 \\ 0 && 1 && 0 \\ \end{smallmatrix}\right)$ and the partition $\mathcal{V}=\{[0,\sfrac{1}{6}),[\sfrac{1}{6},\sfrac{1}{2}),[\sfrac{1}{2},1]\}$ of $[0,1]$. This partition induces the vector of weights $\hat{p}=\left(\frac{1}{6},\frac{1}{3},\frac{1}{2}\right)$.
Since $\widehat{W}$ represents probabilities we need to transform it in one that represents densities with respect to the measure induced by $\hat{p}$. Define $\widehat{W}^\prime=[w^\prime_{ij}]$ by component-wise division of each row of $\widehat{W}$ by $\hat{p}$.  
Therefore, $\widehat{W}^\prime=\left(\begin{smallmatrix}0/\frac{1}{6} && 1/\frac{1}{3} && 0/\frac{1}{2} \\ 0/\frac{1}{6} && \frac{1}{2}/\frac{1}{3} && \frac{1}{2}/\frac{1}{2} \\ 1/\frac{1}{6} && 0/\frac{1}{3} && 0/\frac{1}{2} \\ \end{smallmatrix}\right)=\left(\begin{smallmatrix}0 && 3 && 0 \\ 0 && 3/2 && 1 \\ 6 && 0 && 0 \\ \end{smallmatrix}\right)$. 
Finally, define the DiKernel as $W_\mathcal{V}(x,y)=w^\prime_{ij}$ if $(x,y)\in {\mathcal{V}}^i\times \mathcal{V}^j$ as shown in Figure~\ref{fig:disc_cont_weight_degroot_V}. Note that the procedure of defining the intermediate density matrix $\widehat{W}^\prime$ is equivalent to define $W_\mathcal{V}$ following Definition~\ref{def:discrete_cont_Dikernel}.
\end{example}

\begin{figure}[h!]
    \centering
    \begin{subfigure}{0.5\linewidth}
    \begin{tikzpicture}
        \begin{axis}[width=0.95\linewidth,grid=both,axis lines=box,y dir = reverse, xlabel={$x$}, ylabel={$y$}, xlabel style={below right}, ylabel style={above left}, ytick={0,0.33,0.66,1}, yticklabels = {$0$,$\frac{1}{3}$,$\frac{2}{3}$ ,$1$}, xtick={0,0.33,0.66,1}, xticklabels = {$0$,$\frac{1}{3}$,$\frac{2}{3}$ ,$1$}, ymin=0, ymax=1, xmin=0, xmax=1]
        
        \addplot[fill=black!90, draw opacity=0, fill opacity=0.5]table{
        0 0.33
        0 1
        0.33 1
        0.33 0.33
        0 0.33
        };

        \addplot[fill=black!95, draw opacity=0, fill opacity=0.8]table{
        0.33 0
        0.33 0.33
        0.66 0.33
        0.66 0
        0.33 0
        };

        \addplot[fill=black!95, draw opacity=0, fill opacity=0.8]table{
        0.66 0.33
        0.66 0.66
        1 0.66
        1 0.33
        0.66 0.33
        };
        
        \node at (0.165,0.165) {\Large 0};
        \node at (0.495,0.165) {\Large  \color{white} 3};
        \node at (0.825,0.165) {\Large 0};

        \node at (0.165,0.495) {\Large $\frac{3}{2}$};
        \node at (0.495,0.495) {\Large 0};
        \node at (0.825,0.495) {\Large  \color{white} 3};

        \node at (0.165,0.825) {\Large $\frac{3}{2}$};
        \node at (0.495,0.825) {\Large 0};
        \node at (0.825,0.825) {\Large 0};
        \end{axis}
    \end{tikzpicture}
    \caption{Using partition $\mathcal{P}_3$}
    \label{fig:disc_cont_weight_degroot_P}
    \end{subfigure}\begin{subfigure}{0.5\linewidth}
    \begin{tikzpicture}
        \begin{axis}[width=0.95\linewidth,grid=both,axis lines=box,y dir = reverse, xlabel={$x$}, ylabel={$y$}, xlabel style={below right}, ylabel style={above left}, ytick={0,0.16,0.5,1}, yticklabels = {$0$,$\sfrac{1}{6}$,$\sfrac{1}{2}$ ,$1$}, xtick={0,0.16,0.5,1}, xticklabels = {$0$,$\sfrac{1}{6}$,$\sfrac{1}{2}$,$1$}, ymin=0, ymax=1, xmin=0, xmax=1]
        
        \addplot[fill=black!90, draw opacity=0, fill opacity=1]table{
        0.5 0
        0.5 0.16
        0.16 0.16
        0.16 0
        0.5 0
        };

        \addplot[fill=black!95, draw opacity=0, fill opacity=0.8]table{
        0.5 0.16
        1 0.16
        1 0.5
        0.5 0.5
        0.5 0.16
        };

        \addplot[fill=black!90, draw opacity=0, fill opacity=0.6]table{
        0.16 0.16
        0.16 0.5
        0 0.5
        0 0.16
        0.16 0.16
        };

        \addplot[fill=black!90, draw opacity=0, fill opacity=0.5]table{    0.16 0.5
        0.16 1
        0 1
        0 0.5
        0.16 0.5
        0.16 1
        };
        
        \node at (0.08,0.08) {\Large 0};
        \node at (0.33,0.08) {\Large \color{white} 6};
        \node at (0.75,0.08) {\Large  0 };

        \node at (0.08,0.33) {\Large $\frac{3}{2}$};
        \node at (0.33,0.33) {\Large 0};
        \node at (0.75,0.33) {\Large \color{white} 3};

        \node at (0.08,0.75) {\Large 1};
        \node at (0.33,0.75) {\Large 0};
        \node at (0.75,0.75) {\Large 0};
        \end{axis}
    \end{tikzpicture}
    \caption{Using partition $\mathcal{V}$}
    \label{fig:disc_cont_weight_degroot_V}
    \end{subfigure}
    \caption{DiKernels defined from $\widehat{W}$ and a given partition.}
    \label{fig:disc_cont_weight_degroot}
\end{figure}

\begin{example}[Opinion Dynamics: Uniform Partition]\label{ex:disc_cont_degroot_dynamics} Let $\widehat{W}=\left(\begin{smallmatrix}0 && 1/2 && 1/2 \\ 1 && 0 && 0 \\ 0 && 1 && 0 \\ \end{smallmatrix}\right)$ be a row-stochastic matrix,  $\hat{f}_0=\left(0.5,0.3,0.8\right)^\prime$  the vector of initial opinions and $\mathcal{P}_3=\left\{\left[0,\sfrac{1}{3}\right),\left[\sfrac{1}{3},\sfrac{2}{3}\right),\left[\sfrac{2}{3},1\right]\right\}$ a partition of $[0,1]$. 
Let us define $W_{\mathcal{P}_3}$ as in Example~\ref{ex:disc_cont_dikernel} and $f_0(x) = {\hat{f}_{{0}_{i}}}$ if $x\in \mathcal{P}_3^i$, following Definition~\ref{def:discrete_cont_Dikernel}. A graphical representation of these functions is depicted in Figure~\ref{fig:disc_cont_weight_degroot_P} and~\ref{fig:disc_cont_function_opinions}, respectively. 

From the dynamics of the discrete DeGroot model we can calculate  $\hat{f}_1= \widehat{W}\hat{f}_0= \left(0.55,0.5,0.3\right)^\prime$. 
On the other hand, to update the function of opinions for the continuous case we compute $f_1(x)=\int_0^1W_{\mathcal{P}_3}(x,y)f_0(y)\mathrm{d}y= \int_{\mathcal{P}_3^1}W_{\mathcal{P}_3}(x,y)0.5\mathrm{d}y+
\int_{\mathcal{P}_3^2}W_{\mathcal{P}_3}(x,y)0.3\mathrm{d}y
+\int_{\mathcal{P}_3^3}W_{\mathcal{P}_3}(x,y)0.8\mathrm{d}y
= \left\{ 
\begin{smallmatrix}
0.55 & \text{if } x\in \mathcal{P}_3^1\\
0.5 &  \text{if } x\in \mathcal{P}_3^2\\
0.3 & \text{if } x\in \mathcal{P}_3^3\\
\end{smallmatrix}\right. $. Note that, $f_1(x)$ is the continuous version of $\hat{f}_1$ since $f_1(x)=\hat{f_1}_i \text{ if } x\in \mathcal{P}_3^i$. 
Note that since the dynamics are defined in a recursive way we can conclude that they will be the same at every stage.
\end{example}

\begin{figure}[h!]
    \centering
    \begin{subfigure}{0.5\linewidth}
    \centering
    \begin{tikzpicture}
        \begin{axis}[width=0.95\linewidth,axis lines=left, xlabel={$x$}, ylabel={$f_t$}, xlabel style={below right}, ylabel style={above left}, ytick={0,0.3,0.5, 0.8, 1}, yticklabels = {$0$,$0.3$,$0.5$, $0.8$, $1$}, xtick={0,0.33,0.66,1}, xticklabels = {$0$,$\frac{1}{3}$,$\frac{2}{3}$ ,$1$}, ymin=0, ymax=1, xmin=0, xmax=1]
        \addplot[color=black]table{
        0 0.5
        0.33 0.5
        0.33 0.3
        0.66 0.3 
        0.66 0.8 
        1 0.8
        };
        
        \addplot[color=black, dashed, thick]table{
        0 0.55
        0.33 0.55
        0.33 0.5
        0.66 0.5 
        0.66 0.3 
        1 0.3
        };

        \end{axis}
    \end{tikzpicture}
    \caption{$f_0$ and $f_1$ (dashed)  under $W_{\mathcal{P}_3}$.}
    \label{fig:disc_cont_function_opinions}
    \end{subfigure}\hfill\begin{subfigure}{0.5\linewidth}
    \centering
    \begin{tikzpicture}
        \begin{axis}[width=0.95\linewidth,axis lines=left, xlabel={$x$}, ylabel={$f_t$}, xlabel style={below right}, ylabel style={above left}, ytick={0,0.16,0.5,1}, yticklabels = {$0$,$\sfrac{1}{6}$,$\sfrac{1}{2}$ ,$1$}, xtick={0,0.16,0.5,1}, xticklabels = {$0$,$\sfrac{1}{6}$,$\sfrac{1}{2}$,$1$}, ymin=0, ymax=1, xmin=0, xmax=1]
        \addplot[color=black]table{
        0 0.5
        0.16 0.5
        0.16 0.3
        0.5 0.3 
        0.5 0.8 
        1 0.8
        };
        
        \addplot[color=black, dashed, thick]table{
        0 0.55
        0.16 0.55
        0.16 0.5
        0.5 0.5 
        0.5 0.3 
        1 0.3
        };

        \end{axis}
    \end{tikzpicture}
    \caption{$f_0$ and $f_1$ (dashed) under $W_{\mathcal{V}}$.}
    \label{fig:disc_cont_weighted_function_opinions}
    \end{subfigure}
    \caption{Dynamic of Opinions: From Discrete to Continuous.}
\end{figure}

\begin{example}[Opinion Dynamics: General Partition]\label{ex:disc_cont_weighted_degroot_dynamics} Let $\widehat{W}=\left(\begin{smallmatrix}0 && 1/2 && 1/2 \\ 1 && 0 && 0 \\ 0 && 1 && 0 \\ \end{smallmatrix}\right)$ be a row-stochastic matrix,  $\hat{f}_0=\left(0.5,0.3,0.8\right)^\prime$  the vector of initial opinions and $\mathcal{V}=\{[0,\sfrac{1}{6}),[\sfrac{1}{6},\sfrac{1}{2}),[\sfrac{1}{2},1]\}$ a partition of $[0,1]$. 
Let us define $W_{\mathcal{V}}$ as in Example~\ref{ex:disc_cont_weight_degroot} and $f_0(x) = {\hat{f}_{{0}_{i}}}$ if $x\in \mathcal{V}^i$, following Definition~\ref{def:discrete_cont_Dikernel}. A graphical representation of these functions is depicted in Figure~\ref{fig:disc_cont_weight_degroot_V} and~\ref{fig:disc_cont_weighted_function_opinions}, respectively. 
From the dynamics of the discrete DeGroot model we can calculate  $\hat{f}_1= \widehat{W}\hat{f}_0= \left(0.55,0.5,0.3\right)^\prime$. 
On the other hand, to update the function of opinions for the continuous case we compute $f_1(x)=\int_0^1W_\mathcal{V}(x,y)f(y)\mathrm{d}y= \int_{\mathcal{V}^1}W_\mathcal{V}(x,y)0.5\mathrm{d}y+
\int_{\mathcal{V}^2}W_\mathcal{V}(x,y)0.3\mathrm{d}y
+\int_{\mathcal{V}^3}W_\mathcal{V}(x,y)0.8\mathrm{d}y
= \left\{ 
\begin{smallmatrix}
0.55 & \text{if } x\in \mathcal{V}^1\\
0.5 &  \text{if } x\in \mathcal{V}^2\\
0.3 & \text{if } x\in \mathcal{V}^3\\
\end{smallmatrix}\right. $. Note that, $f_1(x)$ is the continuous version of $\hat{f}_1$ since $f_1(x)=\hat{f_1}_i \text{ if } x\in \mathcal{V}^i$. 
Note that since the dynamics are defined in a recursive way we can conclude that they will be the same at every stage.
\end{example}

Finally, we can focus on the inverse operation, which will takes us from the space of block-constant DiKernels\footnote{The Stochastic Block Models (SBM) is a particular case of block-constant DiKernel.} to the space of row-stochastic matrices of a DeGroot model augmented by a vector of weight. Let $W$ be a block-constant DiKernel. 
The block structure of $W$ allows us to define a partition $\mathcal{V}\times\mathcal{V}$ of $[0,1]^2$ into cells in which the DiKernel takes a constant value. 
This partition allows us to define a density matrix $\widehat{W}^\prime=[w^\prime_{ij}]$  by setting the $ij$-th entry to be equal to the value of $W$ in the block defined by the element $\mathcal{V}^i\times \mathcal{V}^j$ of the partition of $[0,1]^2$. 
Now, note that the partition $\mathcal{V}\times\mathcal{V}$ of $[0,1]^2$ induces an interval partition $\mathcal{V}$ of $[0,1]$. 
We define the weight vector $\hat{p}$ by letting the weight of agent $i$ to be the mass of agents in the $i$-th element of the partition $\mathcal{V}$, therefore $\hat{p}_i=\lambda(\mathcal{V}^i)$. Finally, we define $\widehat{W}$ by component-wise multiplication of each row of $\widehat{W}^\prime$ with $\hat{p}$.  
Similarly as how we illustrate the map from discrete to continuous in Example~\ref{ex:disc_cont_weight_degroot}, we illustrate the map from continuous to discrete in Example~\ref{ex:cont_disc_weight_degroot}. 


\begin{example}[Block-constant DiKernel to Weighted Discrete]\label{ex:cont_disc_weight_degroot}
Let $W$ be the DiKernel given by Figure~\ref{fig:disc_cont_weight_degroot} and let $\mathcal{V}=\left\{\left[0,\sfrac{1}{6}\right),\left[\sfrac{1}{6},\sfrac{1}{2}\right),\left[\sfrac{1}{2},1\right]\right\}$ be the partition of $[0,1]$ induced by $W$. To  compute a discrete DeGroot model, we need to construct a row-stochastic matrix $\widehat{W}$. First, compute the Lebesgue measure of each element of $\mathcal{V}$ to obtain vector of weights $\hat{p}=\left(\sfrac{1}{6},\sfrac{1}{3},\sfrac{1}{2}\right)$. 
Then, define $\widehat{W}^\prime=\left(\begin{smallmatrix}0 && 3 && 0 \\ 0 && \sfrac{3}{2} && 1 \\ 6 && 0 && 0 \\ \end{smallmatrix}\right)$ from the DiKernel $W$.
Finally, compute $\widehat{W}=\left(\begin{smallmatrix}0 && 1 && 0 \\ 0 && 1/2 && 1/2 \\ 1 && 0 && 0 \\ \end{smallmatrix}\right)$ by component-wise multiplication of each row of $\widehat{W}^\prime$ with $\hat{p}$.
\end{example}

\subsection{From Continuous to Discrete}\label{sec:cont_disc}

In the previous section, we have established a natural way to define a block-constant DiKernel from a matrix and vice versa. Now, we would like to define a matrix from a DiKernel which is not block-constant. 
We do it through a discretization process. This operation will allow us to interpret the convergence between a model with a $n\times n$ row-stochastic matrix (discrete DeGroot) and the limit case (continuous DeGroot).
The idea behind DiKernel discretization is to compute block-averages to induce a block-constant DiKernel that we can transform into a DeGroot matrix, as established in Section~\ref{sec:disc_cont}. Once again, there is no \emph{a priori} uniqueness of this discretization so we need to fix a partition $\mathcal{V}$ of $[0,1]$. This partition induces a block structure over $[0,1]^2$ under which we compute the average of the DiKernel $W$ to define the block-constant DiKernel $W_\mathcal{V}$.\footnote{We write $W_{\mathcal{P}_n}$ as $W_{(n)}$ when considering the uniform partition $\mathcal{P}_n$ of $[0,1]$.} We formalize this procedure in Definition~\ref{def:dikernel_discretization}.
Since DiKernels represent probability densities under the uniform measure, we need to multiply by the number of communities (represented by the number of elements in the partition) to define a row-stochastic matrix that represents the social dynamics of a discrete DeGroot model. To further exemplify this, we introduce Examples~\ref{ex:cont_disc_dikernel},~\ref{ex:cont_disc_dikernel_4}, and~\ref{ex:cont_disc_dikernel_gen}.

\begin{definition}[DiKernel Discretization - Continuous to Discrete]\label{def:dikernel_discretization}
    Let $W$ be a DiKernel and the interval partition $\mathcal{V}=\{V^i\}_{i=1}^n$ of $[0,1]$. We can define the block-constant DiKernel $W_{\mathcal{V}}$, a discretization of $W$, as follows. For every $i,j\in [n]$ and every $(x,y)\in V^i \times V^j$,
\[
W_{\mathcal{P}}(x,y)=\frac{1}{\lambda(V^i)\lambda(V^j)}\int_{V^i\times V^j} W(u,v)\mathrm{d}u \mathrm{d}v,
\]
where $\lambda$ represents the Lebesgue measure.
\end{definition}

\begin{example}[Continuous to Discrete: Uniform partition of size 2]\label{ex:cont_disc_dikernel}
Let $W$ be the DiKernel shown in Figure~\ref{fig:discretization_cont_disc_W} and $\widehat{W}_{(2)}$ its $2\times 2$ discretization into a matrix of a discrete DeGroot model.
To compute $\widehat{W}_{(2)}$, first, consider the partition $\mathcal{P}_2=\left\{\left[0,\sfrac{1}{2}\right),\left[\sfrac{1}{2},1\right]\right\}$ and compute $W_{(2)}$, which is shown in Figure~\ref{fig:discretization_cont_disc_W_V}. 
Then, we can transform this block-constant DiKernel $W_{(2)}$ into a row-stochastic matrix by multiplying by $\frac{1}{2}$. 
Finally, we obtain the desired matrix $\widehat{W}_{(2)}=\left(\begin{smallmatrix}\sfrac{3}{4}&\sfrac{1}{4}\\ \sfrac{1}{4}&\sfrac{3}{4}\end{smallmatrix}\right)$. 
\end{example}

\begin{example}[Continuous to Discrete: Uniform partition of size 4]\label{ex:cont_disc_dikernel_4}
Let $W$ be the DiKernel shown in Figure~\ref{fig:discretization_cont_disc_W} and $\widehat{W}_{(4)}$ its $4\times 4$ discretization into a matrix of a discrete DeGroot model.
To compute $\widehat{W}_{(4)}$, first, consider the partition $\mathcal{P}_2=\left\{\left[0,\sfrac{1}{4}\right),\left[\sfrac{1}{4},\sfrac{1}{2}\right),\left[\sfrac{1}{2},\sfrac{3}{4}\right),\left[\sfrac{3}{4},1\right]\right\}$ and compute $W_{(4)}$, which is shown in Figure~\ref{fig:discretization_cont_disc_W_V4}. 
Then, we can transform this block-constant DiKernel $W_{(4)}$ into a row-stochastic matrix by multiplying by $\frac{1}{4}$. 
Finally, we obtain the desired matrix $\widehat{W}_{(2)}=\left(\begin{smallmatrix}\sfrac{3}{4}&\sfrac{1}{4}\\ \sfrac{1}{4}&\sfrac{3}{4}\end{smallmatrix}\right)$. 
Finally, we obtain the desired matrix $\widehat{W}_{(4)}=\left(\begin{smallmatrix}\sfrac{1}{4}&0&\sfrac{1}{4}&\sfrac{1}{2}\\0& \sfrac{1}{4}&\sfrac{1}{2}&\sfrac{1}{4}\\ \sfrac{1}{4}&\sfrac{1}{2}&\sfrac{1}{4}&0\\ \sfrac{1}{2}&\sfrac{1}{4}&0&\sfrac{1}{4} \end{smallmatrix}\right)$. 
\end{example}

\begin{example}[Continuous to Discrete: General partition]\label{ex:cont_disc_dikernel_gen}
Let $W$ be the DiKernel shown in Figure~\ref{fig:discretization_cont_disc_W} and let $\mathcal{V}=\left\{\left[0,\sfrac{1}{4}\right),\left[\sfrac{1}{4},\sfrac{3}{4}\right),\left[\sfrac{3}{4},1\right]\right\}$ a partition of $[0,1]$.  We want to compute $\widehat{W}_{\mathcal{V}}$ which is $3\times 3$ discretization into a matrix of a weighted discrete DeGroot model.
To compute $\widehat{W}_{\mathcal{V}}$, first, we need to compute $W_{\mathcal{V}}$, which is shown in Figure~\ref{fig:discretization_cont_disc_W_V_gen}, and the vector of weights $\hat{p}=\left(\sfrac{1}{4},\sfrac{1}{2},\sfrac{1}{4}\right)$. 
Then, we can transform this block-constant DiKernel $W_{\mathcal{V}}$ into a row-stochastic matrix by component-wise multiplication by $\hat{p}$. 
Finally, we obtain the desired matrix $\widehat{W}_{\mathcal{V}}=\left(\begin{smallmatrix}\sfrac{1}{4}&\sfrac{1}{4}&\sfrac{1}{2}\\ \sfrac{1}{8}&\sfrac{3}{4}&\sfrac{1}{8}\\\sfrac{1}{2}&\sfrac{1}{4}&\sfrac{1}{4}\end{smallmatrix}\right)$. 
\end{example}

\begin{figure}[h!]
    \centering
    \begin{subfigure}{0.5\linewidth}
    \centering
        \begin{tikzpicture}
        \begin{axis}[width=0.95\linewidth,grid=both,axis lines=box,y dir = reverse, xlabel={$x$}, ylabel={$y$}, xlabel style={below right}, ylabel style={above left}, ytick={0,0.25,0.5,0.75,1}, yticklabels = {$0$,$\sfrac{1}{4}$,$\sfrac{1}{2}$,$\sfrac{3}{4}$ ,$1$}, xtick={0,0.25,0.5,0.75,1}, xticklabels = {$0$,$\sfrac{1}{4}$,$\sfrac{1}{2}$,$\sfrac{3}{4}$ ,$1$}, ymin=0, ymax=1, xmin=0, xmax=1]
        
        \addplot[fill=black!90, draw opacity=0, fill opacity=0.8]table{
        0 0
        0 0.25
        0.25 0
        0 0
        };

        \addplot[fill=black!90, draw opacity=0, fill opacity=0.8]table{
        1 0
        0.75 0
        0 0.75
        0 1
        0.25 1
        1 0.25
        1 0
        };

        \addplot[fill=black!90, draw opacity=0, fill opacity=0.8]table{
        1 1
        1 0.75
        0.75 1
        1 1
        };
        
        \node at (0.5,0.5) {\Large\color{white} 2};
        \node at (0.08,0.08) {\small\color{white} 2};
        \node at (0.92,0.92) {\small\color{white} 2};
        \node at (0.25,0.25) {\Large 0};
        \node at (0.75,0.75) {\Large 0};
        \end{axis}
        \end{tikzpicture}
        \caption{$W$}
        \label{fig:discretization_cont_disc_W}
    \end{subfigure}\hfill\begin{subfigure}{0.5\linewidth}
    \centering
        \begin{tikzpicture}
        \begin{axis}[width=0.95\linewidth,grid=both,axis lines=box,y dir = reverse, xlabel={$x$}, ylabel={$y$}, xlabel style={below right}, ylabel style={above left}, ytick={0,0.5,1}, yticklabels = {$0$,$\sfrac{1}{2}$,$1$}, xtick={0,0.5,1}, xticklabels = {$0$,$\sfrac{1}{2}$,$1$}, ymin=0, ymax=1, xmin=0, xmax=1]
        
        \addplot[fill=black!80, draw opacity=0, fill opacity=0.8]table{
        0 0.5
        0 1
        0.5 1
        0.5 0.5
        0 0.5
        };

        \addplot[fill=black!80, draw opacity=0, fill opacity=0.8]table{
        0.5 0
        0.5 0.5
        1 0.5
        1 0
        0.5 0
        };

        \addplot[fill=black!30, draw opacity=0, fill opacity=0.5]table{
        0.5 0
        0.5 0.5
        0 0.5
        0 0
        0.5 0
        };

        \addplot[fill=black!30, draw opacity=0, fill opacity=0.5]table{
        0.5 0.5
        0.5 1
        1 1
        1 0.5
        0.5 0.5
        };

        \node at (0.25,0.25) {\Large $\frac{1}{2}$};
        \node at (0.25,0.75) {\Large $\frac{3}{2}$};
        \node at (0.75,0.25) {\Large $\frac{3}{2}$};
        \node at (0.75,0.75) {\Large $\frac{1}{2}$};
        \end{axis}
        \end{tikzpicture}
        \caption{$W_{(2)}$}
        \label{fig:discretization_cont_disc_W_V}
    \end{subfigure}
    \begin{subfigure}{0.5\linewidth}
    \centering
        \begin{tikzpicture}
        \begin{axis}[width=0.95\linewidth,grid=both,axis lines=box,y dir = reverse, xlabel={$x$}, ylabel={$y$}, xlabel style={below}, ylabel style={above}, ytick={0,0.25,0.5,0.75,1}, yticklabels = {$0$,${\sfrac{1}{4}}$,${\sfrac{1}{2}}$,${\sfrac{3}{4}}$ ,$1$}, xtick={0,0.25,0.5,0.75,1}, xticklabels = {$0$,${\sfrac{1}{4}}$,${\sfrac{1}{2}}$,${\sfrac{3}{4}}$ ,$1$}, ymin=0, ymax=1, xmin=0, xmax=1]
        
        \addplot[fill=black!45, draw opacity=0, fill opacity=0.8]table{
        0 0.5
        0 0.75
        0.25 0.75
        0.25 0.5
        0 0.5
        };
        \addplot[fill=black!90, draw opacity=0, fill opacity=0.8]table{
        0.25 0.5
        0.25 0.75
        0.5 0.75
        0.5 0.5
        0 0.5
        };
        \addplot[fill=black!90, draw opacity=0, fill opacity=0.8]table{
        0 0.75
        0 1
        0.25 1
        0.25 0.75
        0 0.75
        };
        \addplot[fill=black!45, draw opacity=0, fill opacity=0.8]table{
        0.25 0.75
        0.25 1
        0.5 1
        0.5 0.75
        0.25 0.75
        };

        \addplot[fill=black!45, draw opacity=0, fill opacity=0.8]table{
        0 0
        0 0.25
        0.25 0.25
        0.25 0
        0 0
        };
        \addplot[fill=black!45, draw opacity=0, fill opacity=0.8]table{
        0.25 0.25
        0.25 0.5
        0.5 0.5
        0.5 0.25
        0.25 0.25
        };

        \addplot[fill=black!45, draw opacity=0, fill opacity=0.8]table{
        0.5 0
        0.5 0.25
        0.75 0.25
        0.75 0
        0.5 0
        };
        \addplot[fill=black!90, draw opacity=0, fill opacity=0.8]table{
        0.5 0.25
        0.5 0.5
        0.75 0.5
        0.75 0.25
        0.5 0.25
        };
        \addplot[fill=black!45, draw opacity=0, fill opacity=0.8]table{
        0.75 0.25
        0.75 0.5
        1 0.5
        1 0.25
        0.75 0.25
        };
        \addplot[fill=black!90, draw opacity=0, fill opacity=0.8]table{
        0.75 0
        0.75 0.25
        1 0.25
        1 0
        0.75 0
        };

        \addplot[fill=black!45, draw opacity=0, fill opacity=0.8]table{
        0.5 0.5
        0.5 0.75
        0.75 0.75
        0.75 0.5
        0.5 0.5
        };
        \addplot[fill=black!45, draw opacity=0, fill opacity=0.8]table{
        0.75 0.75
        0.75 1
        1 1
        1 0.75
        0.75 0.75
        };

        \node at (0.125,0.125) {\Large $1$};
        \node at (0.375,0.125) {\Large $0$};
        \node at (0.625,0.125) {\Large $1$};
        \node at (0.875,0.125) {\color{white}\Large $2$};

        \node at (0.125,0.375) {\Large $0$};
        \node at (0.375,0.375) {\Large $1$};
        \node at (0.625,0.375) {\color{white}\Large $2$};
        \node at (0.875,0.375) {\Large $1$};

        \node at (0.125,0.625) {\Large $1$};
        \node at (0.375,0.625) {\color{white}\Large $2$};
        \node at (0.625,0.625) {\Large $1$};
        \node at (0.875,0.625) {\Large $0$};

        \node at (0.125,0.875) {\color{white}\Large $2$};
        \node at (0.375,0.875) {\Large $1$};
        \node at (0.625,0.875) {\Large $0$};
        \node at (0.875,0.875) {\Large $1$};
        
        \end{axis}
        \end{tikzpicture}
        \caption{$W_{(4)}$}
        \label{fig:discretization_cont_disc_W_V4}
    \end{subfigure}\hfill\begin{subfigure}{0.5\linewidth}
    \centering
        \begin{tikzpicture}
        \begin{axis}[width=0.95\linewidth,grid=both,axis lines=box,y dir = reverse, xlabel={$x$}, ylabel={$y$}, xlabel style={below}, ylabel style={above}, ytick={0,0.25,0.75,1}, yticklabels = {$0$,${\sfrac{1}{4}}$,${\sfrac{3}{4}}$ ,$1$}, xtick={0,0.25,0.75,1}, xticklabels = {$0$,${\sfrac{1}{4}}$,${\sfrac{3}{4}}$ ,$1$}, ymin=0, ymax=1, xmin=0, xmax=1]
        
        \addplot[fill=black!45, draw opacity=0, fill opacity=0.8]table{
        0 0
        0 0.25
        0.25 0.25
        0.25 0
        0 0
        };
        \addplot[fill=black!30, draw opacity=0, fill opacity=0.5]table{
        0 0.25
        0 0.75
        0.25 0.75
        0.25 0.25
        0 0.25
        };
        \addplot[fill=black!90, draw opacity=0, fill opacity=0.8]table{
        0 0.75
        0 1
        0.25 1
        0.25 0.75
        0 0.75
        };

        \addplot[fill=black!30, draw opacity=0, fill opacity=0.5]table{
        0.25 0
        0.25 0.25
        0.75 0.25
        0.75 0
        0.25 0
        };
        \addplot[fill=black!80, draw opacity=0, fill opacity=0.8]table{
        0.25 0.25
        0.25 0.75
        0.75 0.75
        0.75 0.25
        0.25 0.25
        };
        \addplot[fill=black!30, draw opacity=0, fill opacity=0.5]table{
        0.25 0.75
        0.25 1
        0.75 1
        0.75 0.75
        0.25 0.75
        };

        \addplot[fill=black!90, draw opacity=0, fill opacity=0.8]table{
        0.75 0
        0.75 0.25
        1 0.25
        1 0
        0.75 0
        };
        \addplot[fill=black!30, draw opacity=0, fill opacity=0.5]table{
        0.75 0.25
        0.75 0.75
        1 0.75
        1 0.25
        0.75 0.25
        };
        \addplot[fill=black!80, draw opacity=0, fill opacity=0.8]table{
        0.75 0.75
        0.75 1
        1 1
        1 0.75
        0.75 0.75
        };

        \node at (0.125,0.125) {\Large $1$};
        \node at (0.5,0.125) {\Large $\frac{1}{2}$};
        \node at (0.875,0.125) {\color{white}\Large $2$};

        \node at (0.125,0.5) {\Large $\frac{1}{2}$};
        \node at (0.5,0.5) {\Large $\frac{3}{2}$};
        \node at (0.875,0.5) {\Large $\frac{1}{2}$};

        \node at (0.125,0.875) {\color{white}\Large $2$};
        \node at (0.5,0.875) {\Large $\frac{1}{2}$};
        \node at (0.875,0.875) {\Large $1$};
        
        \end{axis}
        \end{tikzpicture}
        \caption{$W_{\mathcal{V}}$}
        \label{fig:discretization_cont_disc_W_V_gen}
    \end{subfigure}
    \caption{Discretization of DiKernel W.}
\label{fig:discretization_cont_disc}
\end{figure}

From the previous examples we see that for each partition of $[0,1]$ there is a block-constant DiKernel. Now, an important question that rises is how can we compare the opinion dynamics that these two objects generate? Moreover, is there a way to connect these two opinion dynamics? 
To answer this question lets consider a DiKernel $W$, a partition $\mathcal{V}=\{V^i\}_{j\in [J]}$ of $[0,1]$ and a function of opinions $f(x)$. Construct the DiKernel $W_\mathcal{V}$ as the discretization of $W$ under $\mathcal{V}$,  following Definition~\ref{def:dikernel_discretization}. 
Since $W$ and $W_\mathcal{V}$ are different, When we update $f$ under $W$ it differs from the update of $f$ under $W_\mathcal{V}$. We show that this difference is proportional to how different are the two DiKernels. 
Now, if we continue with this updating process these two dynamics will diverge because the errors are being accumulated.  
Nevertheless, since at each stage the error we are making in the update is proportional to the distance between $W$ and $W_\mathcal{V}$ the total error up to stage $t$ is linear in this distance. We formalize this result in Proposition~\ref{prop:dynamic_partition}. 
Note that the fact that the errors are accumulating is no surprise and it is a direct result of considering a dynamic evolution of opinions. 
Nevertheless, this accumulation is well behaved in the sense that it grows at a linear pace.

\begin{proposition}
\label{prop:dynamic_partition} Let $W$ a DiKernels that take values on $[0,M]$ and $\mathcal{V}=\{V^j\}_{j\in [J]}$ a partition of $[0,1]$. Then for each $t\in\mathbb{N}_+$ and for all functions of opinions $f$,
$${\left\|\left(\operatorname{T}^t(W)f\right)(x)-\left(\operatorname{T}^t(W_\mathcal{V})f\right)(x)\right\|}_1\leq 4t\|W-W_\mathcal{V}\|_\square$$
\end{proposition}

In economics, decision makers usually receive utility streams at different moments based on their decisions at a previous stage. 
When taking a decision, economic agents compare their utility streams via discounted evaluation of them to accurately compare the value of money or benefits received today versus those received in the future. 
Therefore, in order to control the linear accumulation of errors we propose a discounted evaluation of the opinions via an utility function. 
We consider that the type of economic agent that gets utility from opinions could be firms, a social-planner and in general decision makers. 
Although, in real-life scenarios opinions are intangible objects from which we could not obtain an utility we consider the utility functions as a way of aggregating these opinions. For example, we can aggregate opinions to compute the average opinion. 
Moreover, if we associate opinions to opposing products or consumption goods, this average can measure market concentration or how likely is the population to consume from a certain good.
To be able to bound the distance between the evaluations of the opinion dynamics coming from the DiKernel $W$ and the DiKernel $W_\mathcal{V}$, we impose Assumption~\ref{assump:lipschitz_payoff} which states Lipschitz conditions over the evaluation function $u$. 
Then, we can conclude that if the evaluation function $u$ satisfy this assumption, the distance between the discounted evaluations of the opinion dynamics is bounded by the distance between the DiKernels that rule the dynamics. We formalize this result in Proposition~\ref{prop:discounted_partition}.

\begin{assumption}[Lipschitz Utility]\label{assump:lipschitz_payoff} 
The utility function $u$ is $\alpha$-Lipschitz. 
\end{assumption} 

\begin{proposition}
\label{prop:discounted_partition}
Let $W$ a DiKernel that takes values on $[0,M]$ and $\mathcal{V}=\{V_j\}_{j\in [J]}$ a partition of $[0,1]$. Consider a utility function $u$ that satisfies Assumption~\ref{assump:lipschitz_payoff}, one initial opinion functions $f\in L^1([0,1])$ and discount factor $\delta\in(0,1)$. Let $\left(f_t\right)_{t\in\mathbb{N}}$ be the dynamic of opinions under $W$ from $f_0$, and let $\left(\tilde{f}_t\right)_{t\in\mathbb{N}}$ be the dynamic of opinions under $W_\mathcal{V}$ from $f_0$.
Then, these dynamics are such that 
$$\left|\sum_{t=1}^{\infty} \delta^t u(f_t)-\sum_{t=1}^{\infty} \delta^t u(\tilde{f}_t)\right|\leq \frac{4\alpha\delta}{(1-\delta)^2}\|W-W_\mathcal{V}\|_\square.$$ 
\end{proposition}

We can also interpret our result from Proposition~\ref{prop:discounted_partition} in the other direction. 
This means that if we start from a DiKernel $W$ that is regular enough, one can choose that it is possible to reach arbitrarily small errors. 
We exploit this idea to establish a convergence result between the discrete and continuous dynamics. 
First, we impose some additional structure on the DiKernel. 
A convenient assumption is for the DiKernel to be Lipschitz continuous. 
But, this excludes block-constant DiKernels. 
Therefore, we introduce the notion of Piecewise Lipschitz Continuity, which allows us to state Theorem~\ref{thm:dynamic_convergence}, the main result of this section.

\begin{assumption}[Piecewise Lipschitz Continuity]\label{assump:piecewise_lipschitz}
The DiKernel $W$ is piecewise $\theta$\nobreakdash-Lipschitz continuous\footnote{A function $W$ that is $\theta$\nobreakdash-Lipschitz is also piecewise $\theta$\nobreakdash-Lipschitz for every partition of $[0,1]$}. That is, there exists $\theta>0$ and a partition $\mathcal{I}=\{I_k\}_{k\in[K]}$ of $[0,1]$ such that for every $(x,y), (x^\prime,y^\prime)\in I_{k_1}\times I_{k_2}$ we have that  $|W(x,y)-W(x^\prime,y^\prime)|\leq \theta\left(|x-x^\prime|+|y-y^\prime|\right)$.
\end{assumption}

\begin{theorem}\label{thm:dynamic_convergence} 
Let $\eta>0$. 
Let $f$ be a function of opinions, and $W$ a DiKernel that satisfies Assumption~\ref{assump:piecewise_lipschitz} for partition $\mathcal{I}$ of $[0,1]$ into $K$ elements. Let $\left(f_t\right)_{t\in\mathbb{N}}$ be the dynamic of opinions under $W$ from $f_0$, and let $\left(\tilde{f}_t\right)_{t\in\mathbb{N}}$ be the dynamic of opinions under $W_\mathcal{V}$ from $f_0$.

Then, there exists a positive integer $n_0$ such that for every $n\geq n_0$, there is a partition $\mathcal{V}=\{V_j\}_{j\in [J]}$ of $[0,1]$ into intervals of measure less than $\frac{1}{n}$ such that for all $t\in\mathbb{N}_+$,  $${\left\|f_t(x)-\tilde{f}_t(x)\right\|}_1\leq \min\left\{2,t\eta\right\}$$
and if $u$ is an utility function that satisfies Assumption~\ref{assump:lipschitz_payoff},
\[
\left|\sum_{t=1}^{\infty} \delta^t u(f_t)-\sum_{t=1}^{\infty} \delta^t u(\tilde{f}_t)\right|\leq \frac{4\alpha\delta}{(1-\delta)^2}\eta.
\]
\end{theorem}


The bound established in Theorem~\ref{thm:dynamic_convergence} increases linearly with time until the upper bound of $2$, which is the diameter of the space of opinions.
Therefore, for a given number of communication stages between the agents we obtain convergence. 
The intuition for this result can be stated as follows.
Given a length $t$ of this communication process, if the mesh of the partition is sufficiently small then the vectors of opinions from the continuous dynamics and its discretization will be close.
It is worth mentioning that this convergence result does not give a way of constructing the optimal partition $\mathcal{V}$ for the discretization.
This is a direct consequence of the lack of uniqueness of the discretization process under different partitions. 
Instead, our result only gives a minimal number of  elements a partition such that the bound is satisfied. 
This provides an informed decision maker with the freedom to choose a partition that effectively bounds the distance between the opinion dynamics.
Moreover, similarly to our intuition for Proposition~\ref{prop:discounted_partition} if the decision maker evaluate opinions through a Lipschitz function $u$ we can bound the accumulation of errors. We extend this analysis for the context of lobby competition in Section~\ref{sec:competition}.

\subsection{Reduction of Dimensionality}\label{sec:reduc_dimension}

One immediate application of the previous section is the construction of a canonical way to reduce the dimensionality of discrete DeGroot models.
We start by considering a classic DeGroot model with a large number of agents. We can interpret this DeGroot model as a block-constant DiKernel with the uniform partition. Since this new object is a DiKernel, we can discretize it by considering groups of agents. 
We can think of this aggregation of agents as considering a coarser partition which is compatible with the uniform partition. 
Under this coarser partition we obtain a new constant-block DiKernel with less parameters that approximate the original DiKernel. This last DiKernel being block-constant, it can be interpreted back as a discrete DeGroot models in smaller dimension. 
Following this intuition we say that we reduce the dimensionality of the DeGroot models in the sense that now we can avoid computing the opinion dynamics in the original matrix and instead compute it in the reduced matrix. 
Our process of dimensionality reduction is depicted in Figure~\ref{fig:dimensionality_reduction} and illustrated in Example~\ref{ex:dimensionality_reduction}.
Note that these tools have two main consequences. 
First, they allow us to quantify the bounds between the opinions coming from two DeGroot models of different dimensions. 
Second, they allow us to compare matrices of different dimensions by analyzing them as objects in the space of bounded measurable functions. Both contributions rely on the distance between continuous objects and their discretizations, which are direct consequences of the results established in Section~\ref{sec:cont_disc}.
We formalize this result in Corollary~\ref{cor:dimensionality_reduction}

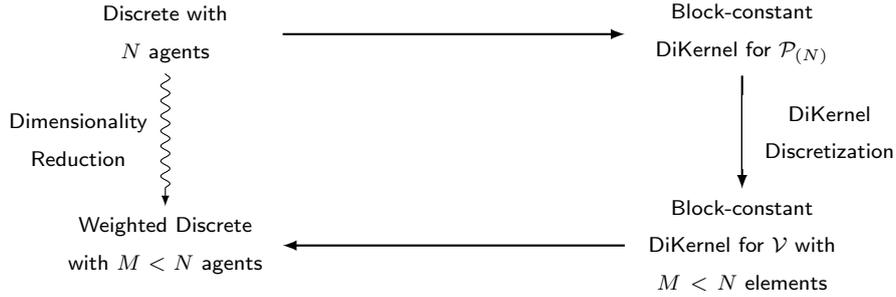
\begin{figure}[H]
    \centering
    \begin{tikzpicture}[every node/.style={font=\sffamily\small, align = center, text width = 3.5cm}]
      \node(C1){Discrete with $n$ agents};
      \node(C2)[right = 5cm of C1] {Block-constant DiKernel for $\mathcal{P}_{(n)}$};
      
      \node(C3)[below = 1.5cm of C2] {Block-constant DiKernel for $\mathcal{V}$ with $m<n$ elements};

      \node(C4)[left = 5cm of C3] {Discrete with $m<n$ groups};

      \path[every node/.style={font=\sffamily\footnotesize}]
      (C1) edge[-latex,thick] (C2)
      edge [-latex,decorate,decoration={snake,amplitude=.6mm,segment length=2mm,post length=2mm}] node(imn)[left ,text width = 2.5cm, fill=white, opacity=0, text opacity=1, align=center]{Dimensionality Reduction} (C4) 
      (C2) edge[-latex,thick] node(imm)[right ,text width = 2.2cm, fill=white, align=center]{DiKernel Discretization} (C3)
      (C3) edge[-latex,thick] (C4);
    \end{tikzpicture}
    \caption{Diagram of the process of dimensionality reduction.}
    \label{fig:dimensionality_reduction}
\end{figure}

\begin{corollary}[Dimensionality Reduction]\label{cor:dimensionality_reduction}
    Let $\widehat{W}$ be a $n\times n$ row-stochastic matrix and $M$ be a partition of $[n]$. Then for each $t\in \mathbb{N}_+$ and every vector of opinions $\hat{f}$, 
    $${\left\|\left(\operatorname{T}^t(W_{(n)})f\right)(x)-\left(\operatorname{T}^t(W_\mathcal{V})f\right)(x)\right\|}_1\leq 4t\|W_{(n)}-W_\mathcal{V}\|_\square,$$ where $W_{(n)}$ is the DiKernel obtained from $\widehat{W}$, $f$ is the opinions function obtained from $\hat{f}$, $\mathcal{V}$ is the partition of $[0,1]$ induced by $M$, and $W_\mathcal{V}$ is the discretization of $W_{(n)}$ under partition $\mathcal{V}$.
\end{corollary}

\begin{example}[Dimensionality Reduction] \label{ex:dimensionality_reduction}
Consider a set of 6 agents interacting according to the social network defined by the row-stochastic matrix 
\[
\widehat{W}={ \left(\begin{smallmatrix}
    0 && 1/2 && 1/2 && 0 && 0 && 0\\
    1 && 0 && 0 && 0 && 0 && 0\\
    0 && 0 && 0 && 1/3 && 1/3 && 1/3\\
    0 && 1/4 && 1/4 && 0 && 1/4 && 1/4\\
    0 && 0 && 0 && 0 && 0 && 1\\
    0 && 1/4 && 1/4 && 1/4 && 1/4 && 0 
\end{smallmatrix}\right)}.
\]
We would like to group these 6 agents into 3 groups of 1, 2, and 3 agents, respectively. 
To do this reduction of dimensionality we first define the DiKernel $W$ (represented in Figure~\ref{fig:Dikernel_society}) that represents the matrix $\widehat{W}$ for the equipartition.
Second, we can define the DiKernel $W_\mathcal{V}$ as the discretization of $W$ with respect to partition $\mathcal{V}=\{[0,\sfrac{1}{6}),[\sfrac{1}{6},\sfrac{1}{2}),[\sfrac{1}{2},1]\}$. 
The DiKernel $W_\mathcal{V}$ is depicted in Figure~\ref{fig:Dikernel_society_groups}.
Then, we can reverse the process to obtain the 
row-stochastic matrix that represents the social network of the 3 groups 
\[
\widehat{W_\mathcal{V}}= \left(\begin{smallmatrix}0 && 1 && 0 \\ 1/2 && 0 && 1/2\\ 0 && 1/3 && 2/3\end{smallmatrix}\right).
\]
augmented by the vector of weights $\hat{p}=\left(\sfrac{1}{6},\sfrac{1}{3},\sfrac{1}{2}\right)$
\end{example}

\begin{figure}[H]
    \centering
    \begin{subfigure}{0.5\linewidth}
    \begin{tikzpicture}
        \begin{axis}[width=0.95\linewidth,grid=both,axis lines=box,y dir = reverse, xlabel={$x$}, ylabel={$y$}, xlabel style={below right}, ylabel style={above left}, ytick={0,0.16,0.33,0.5,0.66, 0.83,1}, yticklabels = {$0$,$\sfrac{1}{6}$,$\sfrac{2}{6}$ ,$\sfrac{3}{6}$,$\sfrac{4}{6}$,$\sfrac{5}{6}$,$1$}, xtick={0,0.16,0.33,0.5,0.66, 0.83,1}, xticklabels = {$0$,$\sfrac{1}{6}$,$\sfrac{2}{6}$ ,$\sfrac{3}{6}$,$\sfrac{4}{6}$,$\sfrac{5}{6}$,$1$}, ymin=0, ymax=1, xmin=0, xmax=1]
        
        \addplot[fill=black!90, draw opacity=0, fill opacity=1]table{
        0.33 0
        0.33 0.16
        0.16 0.16
        0.16 0
        0.33 0
        };

        \addplot[fill=black!95, draw opacity=0, fill opacity=0.8]table{
        0 0.16
        0 0.5
        0.16 0.5
        0.16 0.16
        0 0.16
        };

        \addplot[fill=black!90, draw opacity=0, fill opacity=0.8]table{
        0.33 0.5
        0.33 1
        0.5 1
        0.5 0.5
        0.33 0.5
        };

        \addplot[fill=black!90, draw opacity=0, fill opacity=0.6]table{
        0.5 0.16
        0.5 0.5
        0.66 0.5
        0.66 0.16
        0.5 0.16
        };

        \addplot[fill=black!90, draw opacity=0, fill opacity=0.6]table{
        0.5 0.66
        0.5 1
        0.66 1
        0.66 0.66
        0.5 0.66
        };

        \addplot[fill=black!90, draw opacity=0, fill opacity=1]table{
        0.66 0.83
        0.66 1
        0.83 1
        0.83 0.83
        0.66 0.83
        };
        
        \addplot[fill=black!90, draw opacity=0, fill opacity=0.6]table{
        0.83 0.16
        0.83 0.83
        1 0.83
        1 0.16
        0.83 0.16
        };
        
        \node at (0.08,0.08) {\Large 0};
        \node at (0.24,0.08) {\Large  \color{white} 6};
        \node at (0.66,0.08) {\Large 0};

        \node at (0.08,0.33) {\Large \color{white} 3};
        \node at (0.33,0.33) {\Large 0};
        \node at (0.58,0.33) {\Large $\frac{3}{2}$};
        \node at (0.91,0.5) {\Large  $\frac{3}{2}$};
        \node at (0.75,0.5) {\Large 0};

        \node at (0.16,0.75) {\Large 0};
        \node at (0.41,0.75) {\Large \color{white} 2};
        \node at (0.58,0.58) {\Large 0};
        \node at (0.58,0.83) {\Large $\frac{3}{2}$};
        \node at (0.74,0.91) {\Large \color{white} 6};
        \node at (0.91,0.91) {\Large 0};
        \end{axis}
    \end{tikzpicture}
    \caption{DiKernel $W$ defined from $\widehat{W}$.}
    \label{fig:Dikernel_society}
    \end{subfigure}\hfill\begin{subfigure}{0.50\linewidth}
    \begin{tikzpicture}
        \begin{axis}[width=0.95\linewidth,grid=both,axis lines=box,y dir = reverse, xlabel={$x$}, ylabel={$y$}, xlabel style={below right}, ylabel style={above left}, ytick={0,0.16,0.5,1}, yticklabels = {$0$,$\sfrac{1}{6}$,$\sfrac{1}{2}$ ,$1$}, xtick={0,0.16,0.5,1}, xticklabels = {$0$,$\sfrac{1}{6}$,$\sfrac{1}{2}$,$1$}, ymin=0, ymax=1, xmin=0, xmax=1]

        \addplot[fill=black!95, draw opacity=0, fill opacity=0.8]table{
        0 0.16
        0 0.5
        0.16 0.5
        0.16 0.16
        0 0.16
        };

        \addplot[fill=black!95, draw opacity=0, fill opacity=0.8]table{
        0.16 0
        0.16 0.16
        0.5 0.16
        0.5 0
        0.16 0
        };

        \addplot[fill=black!90, draw opacity=0, fill opacity=0.5]table{    0.16 0.5
        0.16 1
        0.5 1
        0.5 0.5
        0.16 0.5
        };

        \addplot[fill=black!90, draw opacity=0, fill opacity=0.5]table{    0.5 0.16
        0.5 0.5
        1 0.5
        1 0.16
        0.5 0.16
        };

        \addplot[fill=black!90, draw opacity=0, fill opacity=0.7]table{    0.5 0.5
        0.5 1
        1 1
        1 0.5
        0.5 0.5
        };
        
        \node at (0.08,0.08) {\Large 0};
        \node at (0.33,0.08) {\Large \color{white} 3};
        \node at (0.75,0.08) {\Large 0};

        \node at (0.08,0.33) {\Large \color{white} 3};
        \node at (0.33,0.33) {\Large 0};
        \node at (0.75,0.33) {\Large 1};

        \node at (0.08,0.75) {\Large 0};
        \node at (0.33,0.75) {\Large 1};
        \node at (0.75,0.75) {\Large $\frac{4}{3}$};
        \end{axis}
    \end{tikzpicture}
    \caption{DiKernel $W_\mathcal{V}$.}
    \label{fig:Dikernel_society_groups}
    \end{subfigure}
    \caption{Dimensionality Reduction}
\end{figure}
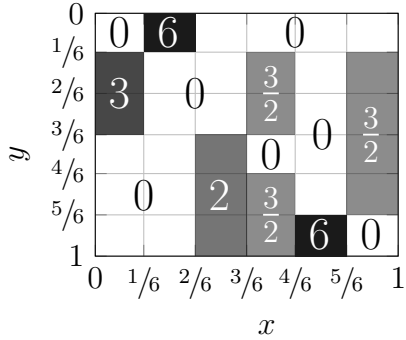
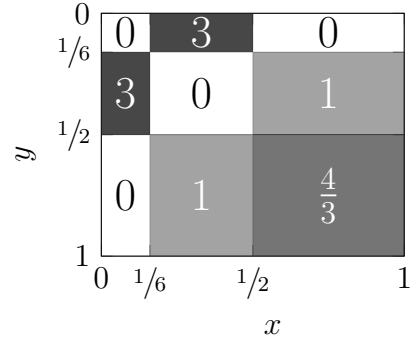

\section{Strategic Influence over Continuous Social Networks}\label{sec:competition}

This section examines a model of influence over continuous social networks. 
We begin by describing a game between two lobbies seeking to influence a continuum of \emph{non-strategic} agents. 
The lobbies aim to bias the agents' initial opinion before letting them update their opinion \textit{\`a la} DeGroot, we call this scenario the ``Lobby Game''. 
We acknowledge that lobbies could exert their influence in a dynamic way by changing the opinion of the agents at different stages during the game. Nevertheless, we focus on the analysis of the static game and leave the dynamic cases as possible future extensions of our model. 

Our analysis begins by examining the set of Nash Equilibria. 
We establish a general existence result and provide examples inspired by the contest literature. 
Then, we characterize the equilibrium for a family of DiKernels and competition operators.
We focus on DiKernels that are defined by a social structure in which all agents assign weights to the influence of others in the same way. 
Finally, to connect our model to real-world economic applications, we introduce the concept of game discretization and study the relationship between the equilibrium of the Lobby Game and its discretized counterpart.

\subsection{The Lobby Game}

We consider a game between two lobbies with opposite interest. 
They compete to bias the opinion of a continuum of non-strategic agents. 
We assume that each lobby chooses its strategy to change the initial opinion of the agents who then interact to update their opinion \textit{\`a la} DeGroot. Moreover, the lobbies get utility from the opinions of the agents at each stage of the DeGroot updating process which are then aggregated by a discounted sum. We focus on a model where this change of initial opinions is local, which reflects the fact that the lobbies choose an influence for each agent.

More formally, let us index the lobbies by $i\in I=\{1,2\}$ and let us assume that Lobby~$1$ aims for the maximal opinion of $1$, whereas Lobby~$2$ aims for the minimal opinion of $-1$.
The agents are indexed by $x\in[0,1]$, they have an initial opinion function $f_0(x)\in[-1,1]$ and they interact according to the DiKernel $W$. 
The timing of the model is as follows: 
\begin{enumerate}[label=\roman*)]
    \item Each lobby $i$ choose a strategy function $s_i$ from its strategy set $S_i$.
    
    We assume that each lobby has a limited budget $b_i$ that he can allocate to influence the agent locally. Formally, a strategy of $i$ is a function $s_i(x)$ in a set $S_i$ such that $\int_0^1s_i(x)\leq b_i$.
    \item The initial function of opinions $f_0$ is updated into a new one defined by the competition operator $\operatorname{C}:L^1([0,1])\times S_1\times S_2\to L^1([0,1])$, where $\operatorname{C}$ is a given operator. 
    
    The competitive operator is defined pointwise. We assume that  the agents have a non\nobreakdash-negative parameter $s_0(x)\in \mathbb{R}_+^*$. 
    This parameter can be interpreted as how sensitive are the agents to the influence of the lobbies, such that the smaller the value of $s_0(x)$, the more sensitive the agents are.  The new opinion of agent $x$ denoted $C(f_0,s_1,s_2)(x)$ is a function of $f_0(x)$,$s_1(x)$, $s_2(x)$ and $s_0(x)$. We will not denote the dependence in $s_0$ since for a given game it will be fixed.

    \item The continuum of agents update their opinion \textit{\`a la} DeGroot in discrete time according to the operator $\operatorname{T}(W)$. 
    \item Utility of Lobby $i$ is given by the utility operator $U_i:L^1([0,1])\times S_1\times S_2\to \mathbb{R}$ where $U_i(f_0,s_1,s_2)=(1-\delta)\sum_{t=1}^{\infty} \delta^t u_i\Bigl(\operatorname{T}^t(W)C(f_0,s_1,s_2)\Bigl)$, where $u_i$ is the stage utility of Lobby~$i$  and $\delta\in(0,1)$ is the discount factor. 
    
    The stage utilities of the lobbies are functions that transform opinions into numbers that the lobbies can interpret as payoff. For example, how close agents are to the Lobby's preferred opinion could reflect intention to consume the product that the Lobby is promoting. We assume that the lobbies only care about the average opinion in society. Let us define, for a given function of opinion $f$, the stage utilities of the lobbies as $u_1(f)=\int_0^1f(x)\psi_1(x)\mathrm{d}x$ and $u_2(f)=-\int_0^1f(x)\psi_2(x)\mathrm{d}x$, where each $\psi_i$ is a non-negative real-valued function. Note that the stage utility of player $i$ is a weighted average of the function of opinions $f$ by $\psi_i$, which could be interpreted as how important is each agent to Lobby $i$. Note that if $\psi_1=\psi_2$ the game becomes a zero-sum game. Moreover, note that the negative sign present in $u_2$ reflects the fact that Lobby~$2$ prefers an average opinion of $-1$. Since the DeGroot model generates a sequence of opinions, we will aggregate these stage utilities by considering the discounted sum of the stage utilities.
\end{enumerate}

Under this framework and timing of the game, we formally define the Lobby Game as in Definition~\ref{def:one_shot}. 
In the rest of the section we study the set Nash Equilibria of the Lobby Game. In particular, we focus on the question of existence, and we study the scenarios under which we can characterize the equilibrium strategies. 
Finally, using the ideas from Sections~\ref{sec:disc_cont} and~\ref{sec:cont_disc}, we study the discretization of the Lobby Game. 
This notion will help us connect our model with real-life economic applications.

\begin{definition}[Lobby Game]\label{def:one_shot} A Lobby Game $\mathcal{G}$ is defined in terms of the set of agents $I=\{1,2\}$, the DiKernel $W$, the Competition Operator $C$, the initial function of opinions $f_0$, the strategy sets $S_i$, the budget $b_i$ the functions $\psi_i$ that determine how important is each agent to the Lobby~$i$, the function $s_0$ that determine how influenciable are the agents, and, the discount factor $\delta$. The normal form of the Lobby Game is $\mathcal{G}(I, S_i, U_i)$. In case of ambiguity, we will include the rest of the parameters of the game as part of the notation. 
\end{definition}

\subsection{Existence of Nash Equilibria}\label{sec:existence}

To study the Nash Equilibria of the game, we establish sufficient conditions on the strategy sets and the competition operator that guarantee the existence of an equilibrium. 
The conditions that we provide are natural in the literature of existence of Nash Equilibria. 
First, for the strategy sets, the Assumption~\ref{assump:strategy_set} impose a structure of convexity and compactness. Second, for the competition operator, the Assumption~\ref{assump:competition_operator} provides the operator with the structure required to solve the optimization problem that the agents are facing. 
These two assumptions, combined, are the sufficient conditions to guarantee the existence of Nash equilibria in the Lobby Game, as established in Theorem~\ref{thm:existence}.
This results follow from \cite{glicksberg1952further}. 
For completeness we provide a proof in the Appendix~\ref{appendix:competition}. 

\begin{assumption}[Strategy Set]\label{assump:strategy_set}
    The strategy sets $S_i$ are convex and compact in a topological vector space.
\end{assumption}
\begin{assumption}[Competition Operator]\label{assump:competition_operator}
    The operator $C:L^1([0,1])\times S_1\times S_2\to L^1([0,1])$ is quasi-concave in $s_1$ and quasi-convex in $s_2$. Moreover, for every $f\in L^1([0,1])$, $C$ is continuous in $s_1$ and $s_2$ in the topology from Assumption~\ref{assump:strategy_set}.
\end{assumption}

\begin{remark}
Let us insist on the fact that the topology ensuring the compacity of the strategy sets (Assumption~\ref{assump:strategy_set}) needs to be compatible with the continuity of the competition operator $C$ (Assumption~\ref{assump:competition_operator}). Since this is necessary for the continuity of the operators $\operatorname{T}(W)$, $u_i$, $u_i\circ \operatorname{T}(W)\circ C$ and $U_i$.
\end{remark}

\begin{theorem}[Existence] \label{thm:existence}
If the game $ \mathcal{G}(I, S_i, U_i, u_i, b_i, s_0, f_0, W, \delta) $ satisfies Assumptions~\ref{assump:strategy_set} and~\ref{assump:competition_operator} then it admits at least one Nash Equilibrium. 
\end{theorem}
\begin{remark}
The existence result relies on the fact that the operator $\operatorname{T}(W)$ is linear. Therefore, we can extend this analysis to general linear updating rules different from the DeGroot model.
\end{remark}

The Lobby Game can be used to model various economic scenarios.  We provide two different scenarios in Examples~\ref{ex:weighted_average_competition} and~\ref{ex:initial_opinion_plus_competition}. 
These examples showcase specific competition operators and strategy sets that satisfy the necessary assumptions, demonstrating the existence of lobby competition games with at least one Nash Equilibrium. One can observe that in both Examples~\ref{ex:weighted_average_competition} and~\ref{ex:initial_opinion_plus_competition}, the two strategy sets represent similar objects, i.e. bounded functions under some norm, but the precise definitions are different. These two spaces are different because the topologies under which they satisfy Assumption~\ref{assump:strategy_set} must be compatible with the topology under which the competition operators satisfy Assumption~\ref{assump:competition_operator}.

\begin{example}\label{ex:weighted_average_competition}
Consider the lobby game with competition operator given by Equation~(\ref{eq:comp_operator_weighted}).
\begin{equation}\label{eq:comp_operator_weighted}
    \operatorname{C}(f_0,s_1,s_2)(x)=\frac{s_1(x)-s_2(x) + f_0(x)s_0(x)}{s_1(x)+s_2(x)+s_0(x)}.
\end{equation}

This operator is inspired from the literature on contests. 
In the standard model of contest, the winner is chosen proportionally to the effort put by each player. 
As it is defined, the new initial opinion, $\operatorname{C}(f_0,s_1,s_2)$ is the expected opinion if the winner imposes its opinion to the agent. 
In addition, let $\mathcal{S}=\{s\text{ such that } s:[0,1]\to \mathbb{R}, \|s\|_\infty\leq \tilde{M} \text{ and } s \text{ is }\text{a Lipschitz function}\}$.
This structure imply that the lobbies cannot exert unbounded influence and that they are influencing society under some regularity constrains\footnote{All the functions in $\mathcal{S}$ have the same Lipschitz constant.}. 

Note that the set $\mathcal{S}$ is convex and compact for the $\|.\|_\infty$-norm (see appendix). The operator $\operatorname{C}$ is concave in $s_1$, convex in $s_2$ and Lipschitz continuous (see appendix). Since they satisfy Assumptions~\ref{assump:strategy_set} and~\ref{assump:competition_operator}, by Theorem~\ref{thm:existence} the Nash Equilibrium of the game exists. 
\end{example}

\begin{example}\label{ex:initial_opinion_plus_competition}
Consider the lobby game with competition operator given by Equation~(\ref{eq:comp_operator_additive}).

\begin{equation}\label{eq:comp_operator_additive}
\operatorname{C}(f_0,s_1,s_2)(x)=\left\{ \begin{smallmatrix}
1 & \text{if } f_0(x)+\frac{s_1(x)-s_2(x)}{s_0(x)}>1\\
-1 & \text{if } f_0(x)+\frac{s_1(x)-s_2(x)}{s_0(x)}<-1\\
f_0(x)+\frac{s_1(x)-s_2(x)}{s_0(x)} &  \text{otherwise}
\end{smallmatrix}\right. .
\end{equation}

This operator is an additive version of the effort given by the two lobbies with cutting when the opinion become too high or too low. 
In addition, let $\mathcal{S}=\{s\text{ such that } s\in L^2([0,1]), \|s\|_2\leq \tilde{M} \}$. 
This structure imply that the lobbies cannot exert unbounded influence although there is no additional constrains on the shape of the strategy functions. 

The appropriate topology is the $\text{weak}^*$ topology from functional analysis in Banach Spaces.
Under this topology, the set $\mathcal{S}$ is convex and compact, and the operator $\operatorname{C}$ is continuous.
The operator $\operatorname{C}$ is monotone, hence it is quasi-concave in $s_1$ and quasi-convex in $s_2$. 
The set $\mathcal{S}$ and the operator $C$ satisfy Assumptions~\ref{assump:strategy_set} and~\ref{assump:competition_operator}. 
Therefore, by Theorem~\ref{thm:existence} the Nash Equilibrium of the game exists. 
\end{example}

\subsection{Equilibrium Characterization}\label{sec:equilibrium_characterization}

The objective of this section is to provide some element of characterization of $\varepsilon$-equilibrium of the Lobby Game. 
We proceed as follows.  We first provide an analysis of the decision problem where only one lobby has a strategic choice and all the agents are influenced by others in the same way (i.e. $W(x,y)=h(y)$). We refer to this family of DiKernels as the Uni-type DiKernel. 
Then, we use this result to define Best-responses for the family of Uni-type DiKernels which allows us to characterize the equilibrium of these games. We conclude by showing that each of this equilibrium is an $\varepsilon$-equilibrium of the game with DiKernel $W$ if the player are sufficiently patient.

\subsubsection{One-player --  Uni-type DiKernel}
We start by focusing on the case in which $W(x,y)=h(y)$. 
Under this assumption, every agent is influenced by society in the same way. 
This property implies that the social network satisfies some anonymity since the labels of the agents can be interchanged (since for all $x,z\in [0,1]$, $W(x,y)=W(z,y)=h(y)$). 
We refer to this family of DiKernel as the Uni-type DiKernel. 
Although, this assumption might appear restrictive, we will relax this assumption later in this section as the Uni-type DiKernels approximates a more general class of DiKernel $W(x,y)$.

We focus on the Uni-type DiKernel because this assumption allows us to conclude that after the first iteration of the model of DeGroot, the agents reach consensus since the function of opinions becomes constant.
Then, the subsequent iterations of the model of DeGroot would not change the opinion thanks to the row-stochastic property of the DiKernel. 
Moreover, since the opinions are constant, the weighting function $\psi_1$ used by the Lobby~1 does not play a strategic role anymore since $\int_0^1\psi_1(x)\mathrm{d}x=1$. 
Therefore, we can write the Lobby's problem as

\begin{align}
    \max_{s_1} & \left[\int_0^1 h(y)\operatorname{C}(f_0,s_1)(y)\mathrm{d}y\right] 
    \label{eq:lobby_problem_general}\\
\text{subject to} & \int s_1(x) \mathrm{d} x \leq b_1 \text{ and } s_1(x)\geq 0\notag
\end{align}

We can further simplify this optimization problem by incorporating the positivity constraint over $s_1(x)$ into the budget constrain via a change of variables by replacing $s_1$ with $\zeta_1^2$. 
The positivity constraint prevents lobbies from negatively influencing agents. 
Without this constraint, a lobby might reduce the opinion of agents already aligned with its preference to influence others more effectively.
Then, we can rewrite the Lobby's problem as
\begin{align}
    \max_{\zeta_1} & \left[\int_0^1 h(y)\operatorname{C}(f_0,\zeta_1^2)(y)\mathrm{d}y\right] \label{eq:lobby_problem_general2}\\ 
\text{subject to} & \int \left(\zeta_1(x)\right)^2 \mathrm{d} x \leq b_1 \notag
\end{align}

To get a closed-form formula for the strategy we need to focus on a particular case of the competition operator. 
Let us consider the one-player restriction of the operator defined in Example~\ref{ex:weighted_average_competition}, namely, $\operatorname{C}(f_0,\zeta_1^2)(x) = \frac{\zeta_1^2(x)+ s_0(x)f_0(x)}{\zeta_1^2(x)+s_0(x)}$. 
In particular, we can rewrite the lobby's objective function as $\int_0^1 \frac{s_0(y)\left[f_0(y)-1\right]h(y)}{\zeta_1^2(y)+s_0(y)}\mathrm{d}y$. To solve this problem, we treat it as a problem of calculus of variations. Define the Lagrangian as $$\mathcal{L}\left(y,\zeta_1^2,{\zeta_1^2}^\prime\right) = \left[\frac{s_0(y)\left[f_0(y)-1\right]h(y)}{\zeta_1^2(y)+s_0(y)}-\mu\left(\zeta_1^2(y)-b_1\right)\right].$$ 

We compute the First Order Conditions with respect to $\zeta_1^2$ and $\mu$, which are given by Equations~(\ref{eq:zeta_positive}), and~(\ref{eq:budget_constrain}).
\begin{align}
&\zeta_1^2(y)=\sqrt{\frac{s_0(y)\left[f_0(y)-1\right]h(y)}{\mu}}-s_0(y)\text{, or } \zeta_1^2(y)=0\text{, and }  \label{eq:zeta_positive} \\
&\int_0^1\zeta_1^2(y)\mathrm{d}y= b_1\label{eq:budget_constrain}
\end{align}

Since the problem is well-defined, we solve the system of equations to identify the agents for whom the lobby's strategy is positive. 
At equilibrium, the lobby uses its entire budget to maximize influence.
By substituting the first part of Equation~(\ref{eq:zeta_positive}) into~(\ref{eq:budget_constrain}) we define an implicit $\mu(x)$ that characterizes the set $A(x,\mu(x))$ such that $\zeta_1\left(A(x,\mu(x)\right)$ is positive, which gives us the characterization of the one\nobreakdash-player case.

\subsubsection{Two-players - Uni-type DiKernel}

In general, a characterization for the two player case is a complicated task.
Nevertheless, we can build on the one player case for the two competition operators that we introduced. Let us focus on the weighted average competition operator to illustrate the procedure.\\
The key idea is to rewrite the operator $$\operatorname{C}(f_0,s_1,s_2)(x) = \operatorname{C}(\breve{f_0},s_1)(x)=\frac{s_1(x) +\breve{s_0}(x)\breve{f_0}(x)}{s_1(x)+\breve{s_0}(x)}$$ where $\breve{s_0}(x)=s_2(x)+s_0(x)$ and $\breve{f_0}(x)=\frac{-s_2(x)+s_0(x)f_0(x)}{s_2(x)+s_0(x)}$.
This has exactly the same structure as the competition operator with only one player. Hence, we can follow the same algebraic manipulations as in the previous section to obtain the best response $BR_1(s_2)$. Similarly, one can compute $BR_2(s_1)$ which allow us to characterize the Nash Equilibrium strategies. 

\subsubsection{From Uni-type DiKernels to $\gamma$-mixing DiKernels}

While the family of Uni-type DiKernels simplifies characterizing equilibrium strategies, we aim for a result applicable to a broader class of DiKernels. Section~\ref{sec:conditions_consensus} examines sufficient conditions under which a DiKernel $W(x,y)$ leads to consensus among agents that is obtained as the weighted average of the initial opinions for some $h(y)$. One can interpret $h$ as a Uni-type DiKernel. Hence, we can exploit this property to utilize the characterization developed for the Uni-type DiKernel.

\begin{corollary}[of Proposition~\ref{prop:existence_unitype}]\label{cor:Dikernel_convergence_markov}
Let $W$ be a DiKernel that satisfies Assumption~\ref{assump:toto_mixing} and $u$ be a utility function. Given $\varepsilon>0$, there exists a large enough discount factor $\delta$ such that the discounted payoff under $W$ from $f_0$ and the discounted payoff under $h$ from $f_0$ are close, where $h$ is determined by Proposition~\ref{prop:existence_unitype}.
\end{corollary}

This result allow us to characterize approximations of the equilibrium strategies for a larger family of DiKernels. Following Corollary~\ref{cor:Dikernel_convergence_markov}, any equilibrium of the Lobby Game with the Uni-type DiKernel $h$ is an $\varepsilon$-equilibrium strategies of the Lobby Game with DiKernel $W$ for large enough discount factor.

\subsection{Convergence of strategies: From Continuous to Discrete Games} \label{sec:convergence_strat}

We have established sufficient conditions that guarantee the existence of the Nash Equilibria of the Lobby Game. Moreover, we have studied the class of DiKernel that we can characterize the equilibrium strategies. 
Then, in this section, we would like  to use these results as tools to find equilibria of large discrete Lobby Games. 
The family of large discrete Lobby Games are games closer to what we can observe in real-life economic scenarios. 
In this class of games, the lobbies sometime face a complex optimization problem in the sense that they need to choose a level of influence for each agent or which set of agents they want to influence. 
Therefore, this section aim to simplify the analysis for this family of finite games. 
We use the tools defined in Section~\ref{sec:cont_disc} to connect the continuous and discrete worlds in this game-theorical scenario. 

We start by defining the discretization of the Lobby Game by replacing the DiKernel $W$ for $W_\mathcal{V}$ letting the rest of the elements unchanged. We describe this discretization\footnote{One could think of alternative discretizations to include the discretization of the initial function of opinions. In this paper we only focus on discretizations of the DiKernel.} process formally in Definition~\ref{def:disc_one_shot}.
The main intuition for this discretization process is that lobbies will continue to influence the continuum of agents which interact according to a block-constant DiKernel. The blocks in the DiKernel represent that agents are part of a community and that there is a finite number of communities. 

\begin{definition}[Discretization of One-Shot Game] \label{def:disc_one_shot} 
Let $\mathcal{G}(I, S_i, U_i)$ be the normal form of the Lobby Game. The discretization of the one-shot game according to partition $\mathcal{V}$ is $\tilde{\mathcal{G}}_\mathcal{V}(I, S_i, U_i)$. As in Definition~\ref{def:one_shot}, in case of ambiguity, we will include the rest of the parameters of the game in the notation.
\end{definition}

We use our technical results from Section~\ref{sec:continuous_social_dynamics} and assume that the stage utility functions of the lobbies satisfy Assumption~\ref{assump:lipschitz_payoff}. 
This fact allows us to bound the distance of the utility functions between the game and its discretization. 
The key point is that if the utility functions are Lipschitz continuous, the linear growth of the error showed in Theorem~\ref{thm:dynamic_convergence} is controlled by the discount factor. 
With these bounds, we state Theorem~\ref{thm:epsilon_nash} which allows us to study the $\varepsilon$\nobreakdash-Nash Equilibria  of the discretization of the game. 
These $\varepsilon$\nobreakdash-Nash Equilibria as defined in Definition~\ref{def:epsilon_nash} are strategies that approximately satisfy the Nash Equilibrium condition.

\begin{definition}[$\varepsilon$\nobreakdash-Nash Equilibrium]\label{def:epsilon_nash}
Consider the Lobby Game $\mathcal{G}(I, S_i, U_i)$. 
The strategy profile $(\sigma_i,\sigma_{-i})\in \prod_{i\in I} S_i$ is an $\varepsilon$\nobreakdash-Nash Equilibrium of the game if for every player $i$, 
$U_i(\sigma_i,\sigma_{-i})\geq U_i(\sigma_i^\prime,\sigma_{-i})-\varepsilon$ for all $\sigma_i^\prime\in S_i$. 
\end{definition}

\begin{theorem}\label{thm:epsilon_nash} Let $\mathcal{G}$ be the game that satisfies Assumptions~\ref{assump:lipschitz_payoff},~\ref{assump:piecewise_lipschitz},~\ref{assump:strategy_set}, and ~\ref{assump:competition_operator}.  For every $\varepsilon>0$, there exists a positive integer $n_0$ such that for every $n\geq n_0$, there is a partition $\mathcal{V}=\{V_j\}_{j\in [J]}$ of $[0,1]$ into intervals of measure less than $\frac{1}{n}$. Let $\tilde{\mathcal{G}}_\mathcal{V}$ be the discretization of the game $\mathcal{G}$ according to partition $\mathcal{V}$. Then, the Nash Equilibria of $\mathcal{G}$ are $\varepsilon$\nobreakdash-Nash Equilibria of $\tilde{\mathcal{G}}_\mathcal{V}$. Conversely, the Nash Equilibria of $\tilde{\mathcal{G}}_\mathcal{V}$ are $\varepsilon$\nobreakdash-Nash Equilibria of $\mathcal{G}$.
\end{theorem}

\section{Concluding Remarks}\label{sec:conclusion}
In this paper, we extend the DeGroot model to a setup with a continuum of agents. 
We show that, under certain regularity conditions, the continuous DeGroot model is the limit case of the discrete one.
Additionally, we establish sufficient conditions under which consensus emerges in the continuous DeGroot model.
This extension, in turn, allows us to study economic scenarios that involve interactions between a large number of agents through their continuous approximation.

We propose applications of our model in both non-strategic and strategic settings.
First, we develop a canonical map between discrete and continuous models, enabling us to control the distance between their dynamics.
This map allows us to tackle the problem of dimensionality reduction by comparing discrete problems of different sizes in their continuous form. 
Second, we focus on lobby competition and state the sufficient conditions for the existence of the Nash Equilibrium. 
Under this game-theoretic setup we provide the characterization when agents interact under a social structure in which everyone is influenced in the same way.
Furthermore, we use our technical results from the convergence between the continuous and discrete models of DeGroot to approximate the equilibrium strategies of the discretizations of the Lobby Game.

\bibliographystyle{apalike} 
\bibliography{biblio}
\begin{appendix}

\section{Supplementary Material for Section~\ref{sec:continuous_social_dynamics}}

We start with some technical results that allow us to characterize the convergence of the dynamic of opinions.
These results establish the equivalence of results between the discrete DeGroot models and our proposed continuous version (see Definition~\ref{def:degroot_cont}). One might think that these two dynamics (discrete and continuous) are not equivalent because there could be an accumulation of errors leading to divergence. Nevertheless, our first technical result, Lemma~\ref{lem:graphon_distance}, bound this difference. 
Specifically, it shows that the distance between updated functions of opinions is bounded by the distance between the initial functions of opinions and the DiKernels that determine agent interaction. 
This result is positive in the sense that, if the initial opinions and DiKernels are well behave then the final opinions are not diverging.

\begin{lemma}\label{lem:graphon_distance} Let $W$ and $V$ DiKernels that take values on $[0,M]$. Then for all functions of opinions $f$ and $g$ the following inequality holds $$ {\left\|\left(\operatorname{T}(W)f\right)(x)-\left(\operatorname{T}(V)g\right)(x)\right\|}_1\leq {\|f-g\|}_1+4{\|W-V\|}_\square$$
Where $\|U\|_{\square}=\sup\limits_{A\times B\subseteq{[0,1]}^2} \left|\int_A\int_B U(x,y)\mathrm{d}x \mathrm{d}y\right|$. 
\end{lemma}
\begin{proof} 
Let $\Delta_F={\left\|\left(\operatorname{T}(V)f\right)(x)-\left(\operatorname{T}(V)g\right)(x)\right\|}_1$ be the error one incurs in the updating different functions of opinions $f$ and $g$, respectively, using a common DiKernel $V$. In addition, let $\Delta_G = {\left\|\left(\operatorname{T}(W)f\right)(x)-\left(\operatorname{T}(V)f\right)(x)\right\|}_1$ be the error one incurs when updating a common function of opinions $f$ using DiKernels $W$ and $V$, respectively. From the triangle inequality of $\|.\|_1$ we know that
$\left\|\left(\operatorname{T}(W)f\right)(x)-\left(\operatorname{T}(V)g\right)(x)\right\|_1 \leq \Delta_F+\Delta_G$, which allows us to explore these elements independently. We start by bounding $\Delta_G$. Then, 
\begin{align*}
    \Delta_G =& {\left\|\int_0^1 W(x,y)f(y)\mathrm{d}y-\int_0^1 V(x,y)f(y)\mathrm{d}y(x)\right\|}_1\\
    =& {\left\|\int_0^1 \left[W(x,y)-V(x,y)\right]f(y)\mathrm{d}y\right\|}_1\\
    =& \int_0^1\left| \int_0^1 \left[W(x,y)-V(x,y)\right]f(y)\mathrm{d}y\right|\mathrm{d}x
\end{align*}
Define $h(x)=\sign\left( \int_0^1 \left[W(x,y)-V(x,y)\right]f(y)\mathrm{d}y\right)$. Rewrite the previous integral by using $h(x)$ instead of the absolute value. 
\begin{align*}
    \Delta_G =& \int_0^1 h(x) \left(\int_0^1 \left[W(x,y)-V(x,y)\right]f(y)\mathrm{d}y\right)\mathrm{d}x\\
    =& \int_0^1\int_0^1 \left[W(x,y)-V(x,y)\right]f(y)h(x)\mathrm{d}x\mathrm{d}y\\
\end{align*}
Consider the auxiliary norm $\|U\|_{\infty\rightarrow 1}=\sup\limits_{-1\leq h,f\leq 1} \int_0^1\int_0^1 U(x,y)h(x)f(y)$, which is related to $\|.\|_\square$ by $\frac{1}{4}\|U\|_{\infty\rightarrow 1}\leq \|U\|_{\square}\leq \|U\|_{\infty\rightarrow 1}$. Then we can conclude that, 
\begin{equation*}
    \Delta_G = \int_0^1\int_0^1 \left[W(x,y)-V(x,y)\right]f(y)h(x)\mathrm{d}x \mathrm{d}y\leq  \| W-V\|_{\infty\rightarrow 1}\leq 4 \| W-V\|_{\square}
\end{equation*}

Now, we can analyze bound $\Delta_F$. Then, 
\begin{align*}
    \Delta_F= & \int_0^1\left|\int_0^1 V(x,y)\left(f(y)-g(y)\right)\mathrm{d}y\right|\mathrm{d}x \leq  \int_0^1\int_0^1 \left|W(x,y)\left(f(y)-g(y)\right)\right|\mathrm{d}y\mathrm{d}x \\
    \leq & \int_0^1\int_0^1 \left|1\left(f(y)-g(y)\right)\right|\mathrm{d}y\mathrm{d}x  \leq  \int_0^1 \left\|f-g\right\|_1\mathrm{d}x=\left\|f-g\right\|_1
\end{align*}

We can conclude that ${\left\|\left(\operatorname{T}(W)f\right)(x)-\left(\operatorname{T}(V)g\right)(x)\right\|}_1 \leq \Delta_F + \Delta_G \leq \left\|f-g\right\|_1 + 4 \| W-V\|_{\square}$, which is the desired result.
\end{proof}

\begin{lemma}\label{prop:dynamic_convergence} Let $W$ and $V$ DiKernels that take values on $[0,M]$. Then for each $t\in\mathbb{N}_+$ and for all functions of opinions $f$ and $g$ the following inequality holds
$${\left\|\left(\operatorname{T}^t(W)f\right)(x)-\left(\operatorname{T}^t(V)g\right)(x)\right\|}_1\leq \|f-g\|_1 + 4t\|W-V\|_\square$$
\end{lemma}
\begin{proof}
    Let $t=1$. From Lemma~\ref{lem:graphon_distance} we know that ${\left\|\left(\operatorname{T}(W)f\right)(x)-\left(\operatorname{T}(V)g\right)(x)\right\|}_1\leq \|f-g\|_1+4 \|W-V\|_{\square}$, which proves that the inequality holds for $t=1$. We now exploit the recursive property of the DeGroot model to prove that the inequality holds for a general $t$. 
    First, let us express $\left(\operatorname{T}^t(W)f\right)(x)$ as $\left(\operatorname{T}(W)\left(T^{t-1}(W)f\right)\right)(x)=\left(\operatorname{T}(W)f^\prime\right)(x)$ and $\left(\operatorname{T}^t(V)g\right)(x)$ as $\left(\operatorname{T}(V)\left(T^{t-1}(V)g\right)\right)(x)=\left(\operatorname{T}(V)g^\prime\right)(x)$. 
    Then, from Lemma~\ref{lem:graphon_distance} and the result for $t=1$ we know that $ \|\left(\operatorname{T}(W)f^\prime\right)(x)-\left(\operatorname{T}(V)g^\prime\right)(x)\|_1 \leq \|f^\prime-g^\prime\|_1+4 \|W-V\|_{\square}$.
    Note that $ \|f^\prime-g^\prime\|_1=\|\left(T^{t-1}(W)f\right)(x)-\left(T^{t-1}(V)f\right)(x)\|_1$ and we can use the same argument to bound $\|\left(T^{t-1}(W)f\right)(x)-\left(T^{t-1}(V)f\right)(x)\|_1$ with the operators in $t-2$. Since we can iterate this process $t$ times it follows that $\|\left(\operatorname{T}^t(W)f\right)(x)-\left(\operatorname{T}^t(W_{\mathcal{V}})f\right)(x)\|_1\leq \|f-g\|_1 + 4t\|W-V\|_\square$.
\end{proof}

This lemma shows that the distance between the two dynamics grows linearly with a speed controlled by $\|W-V\|_\square$. 
In particular, if we study a problem where what happened at time $t$ is discounted we can control the difference.

\begin{proof}[Proof of Proposition~\ref{prop:dynamic_partition}] Consider Lemma~\ref{prop:dynamic_convergence} when $f=g$.
\end{proof}

\begin{lemma}\label{lem:discounted_two_kernels}
Given $W$ and $V$ two DiKernels that take values on $[0,M]$, a stage utility function $u$ that satisfies Assumption~\ref{assump:lipschitz_payoff}, two initial opinion functions $f,g\in L^1([0,1])$ and discount factor $\delta\in(0,1)$. Then, the dynamic of opinions $f_t$ given by $\operatorname{T}^t(W)f$ and $g_t$ given by $\operatorname{T}^t(V)g$, are such that 

$$\left|\sum_{t=1}^{\infty} \delta^t u(f_t)-\sum_{t=1}^{\infty} \delta^t u(g_t)\right|\leq \frac{\alpha\delta}{1-\delta}\|f-g\|_1+\frac{4\alpha\delta}{\left(1-\delta\right)^2}\|W-V\|_\square.$$
\label{lem:convergence_payoff}
\end{lemma}
\begin{proof}
Denote by $f_t=\left(\operatorname{T}^t(W)f\right)$ and $g_t=\left(\operatorname{T}^t(V)g\right)$. Let $\alpha$ be the Lipschitz constant of $u$, given by Assumption~\ref{assump:lipschitz_payoff}.
First, note that 
$\left|\sum_{t=1}^{\infty} \delta^t u(f_t) - \sum_{t=1}^{\infty} \delta^t u(g_t)\right| \leq \sum_{t=1}^{\infty} \delta^t | u(f_t)- u(g_t)|\leq \sum_{t=1}^{\infty} \delta^t \alpha\|f_t-g_t\|_1$.

Now, Lemma~\ref{prop:dynamic_convergence} tell us that $
\|f_t-g_t\|_1\leq \|f-g\|_1 + 4t\|W-V\|_\square$, hence
\[
\left|\sum_{t=1}^{\infty} \delta^t u(f_t) - \sum_{t=1}^{\infty} \delta^t u(g_t)\right| \leq \sum_{t=1}^{\infty} \delta^t | u(f_t)- u(g_t)|\leq \sum_{t=1}^{\infty} \delta^t \alpha\left(\|f-g\|_1 + 4t\|W-V\|_\square\right).
\]
Since $\sum_{t=1}^{\infty} \delta^t=\frac{\delta}{1-\delta}$ and $\sum_{t=1}^{\infty} \delta^t t= \frac{\delta}{\left(1-\delta\right)^2}$, we obtain that
\[
\left|\sum_{t=1}^{\infty} \delta^t u(f_t) - \sum_{t=1}^{\infty} \delta^t u(g_t)\right| \leq \frac{\alpha\delta}{1-\delta}\|f-g\|_1+\frac{4\alpha\delta}{\left(1-\delta\right)^2}\|W-V\|_\square.
\]
\end{proof}

\begin{proof}[Proof of Proposition~\ref{prop:discounted_partition}] Consider Lemma~\ref{lem:discounted_two_kernels} when $f=g$.
\end{proof}

To determine conditions for the convergence of the discrete and the continuous versions of the DeGroot model we focus on the distance between the DiKernels $W$ and $W_{\mathcal{V}}$, where the latter is the discretization of $W$ using the partition $\mathcal{V}$ of $[0,1]$.

In Lemma~\ref{lem:distance_partition2} we prove that the distance between a DiKernel and it discretization converges when considering a finer partition of $[0,1]$, given that the original DiKernel satisfy Assumption~\ref{assump:piecewise_lipschitz}.
Finally, Theorem~\ref{thm:dynamic_convergence} establish the convergence between the continuous DeGroot model that we propose and the discrete version proposed by \cite{degroot1974}.

\begin{lemma}\label{lem:distance_partition2} Let $\mathcal{V}=\{V_j\}_{j\in [J]}$ be a partition of $[0,1]$ into intervals of measure less than $\frac{1}{n}$ and $W$ a DiKernel that satisfies Assumption~\ref{assump:piecewise_lipschitz} then it holds $$\|W-W_\mathcal{V}\|_\square \leq \frac{2\theta}{n}+M\frac{K^2}{n^2}$$
\end{lemma}
\begin{proof}
Since $W$ satisfies Assumption~\ref{assump:piecewise_lipschitz} there exists a partition $\mathcal{I}=\{I_k\}_{k\in [K]}$ of $[0,1]$ into $K$ elements. First note that, in the general case the interval partition $\mathcal{V}$ of $[0,1]$ is not a refinement of $\mathcal{I}$. 
Let $\tilde{I}=\left\{i\in[I]|\exists k\in [K] \text{ such that } V_i\subseteq I_k\right\}$ the set of indices of elements of the partition $\mathcal{V}$ that are subsets of some element of the partition $\mathcal{I}$.  We know that in such elements of  the function is Lipschitz. 
Now, consider the indices $j\in[I]$ such that $V_j$ intersect more than one element of the partition $\mathcal{I}$. 
Since each interval $I_k$ takes is of the form $[a_k,a_{k+1})$ for $k<K$ and $[a_k,1]$ for $k=K$, it is clear that each $a_k$ belongs to a unique element of the partition $\mathcal{V}$. Let $\tilde{J}=\left\{i\in I: a_k\in V_i \text{ for some } k\in K \right\}$ be the set of those indices. The sets of indices $\tilde{I}$ and $\tilde{J}$ allow us to split the problem as shown in Equation~\ref{eq:split_norm}. 

\begin{align}
    \|W-W_{\mathcal{V}}\|_\square & = \sup\limits_{A\times B\subseteq{[0,1]}^2} \left|\int_{A\times B} W(x,y)-W_\mathcal{V}(x,y)\mathrm{d}x \mathrm{d}y\right| \nonumber\\
         & \leq \sup\limits_{A\times B\subseteq{[0,1]}^2} \int_{A\times B }\left|W(x,y)-W_\mathcal{V}(x,y)\right|\mathrm{d}x \mathrm{d}y\nonumber\\ 
         & \leq \int_{[0,1]^2}\left|W(x,y)-W_\mathcal{V}(x,y)\right|\mathrm{d}x \mathrm{d}y\nonumber\\
         &=  \sum_{i,j\in[I]}\ \int_{V_i\times V_j}\left|W(x,y)-W_\mathcal{V}(x,y)\right|\mathrm{d}x \mathrm{d}y\nonumber\\
         &\leq \Delta_{\tilde{I}} + \Delta_{\tilde{J}} \label{eq:split_norm}
\end{align}

We will start by analyzing $\Delta_{\tilde{I}}$ which is the cells where the function is $L$\nobreakdash-Lipschitz. 

\begingroup
\allowdisplaybreaks
\begin{align*}
    \Delta_{\tilde{I}}&=\sum_{i,j\in \tilde{I}}\ \int_{V_i\times V_j}\left|W(x,y)-W_\mathcal{V}(x,y)\right|\mathrm{d}x \mathrm{d}y\\
    & = \sum_{i,j\in \tilde{I}}\ \int_{V_i\times V_j}\left|W(x,y)-\frac{1}{\lambda(V_i)\lambda(V_j)}\int_{V_i\times V_j} W(u,v)\mathrm{d}u \mathrm{d}v\right|\mathrm{d}x \mathrm{d}y\\
         & \leq \sum_{i,j\in \tilde{I}}\ \int_{V_i\times V_j}\frac{1}{\lambda(V_i)\lambda(V_j)}\left[\int_{V_i\times V_j}\left|W(x,y)- W(u,v)\right|\mathrm{d}u \mathrm{d}v\right]\mathrm{d}x \mathrm{d}y\\
         & \leq \sum_{i,j\in \tilde{I}}\ \int_{V_i\times V_j}\frac{1}{\lambda(V_i)\lambda(V_j)}\left[\int_{V_i\times V_j}\theta(|x-u|+|y-v|) \mathrm{d}u \mathrm{d}v\right]\mathrm{d}x \mathrm{d}y\\
         & \leq \sum_{i,j\in \tilde{I}}\ \int_{V_i\times V_j}\frac{1}{\lambda(V_i)\lambda(V_j)}\left[\int_{V_i\times V_j}\frac{2\theta}{n} \mathrm{d}u \mathrm{d}v\right]\mathrm{d}x \mathrm{d}y\\
         & = \frac{2\theta}{n} \sum_{i,j\in \tilde{I}}\ \lambda(V_i)\lambda(V_j) \leq \frac{2\theta}{n} \sum_{i,j\in [I]}\ \lambda(V_i)\lambda(V_j)\\
         & \leq \frac{2\theta}{n}
\end{align*}
\endgroup

Now we will analyze the indices in $\tilde{J}$.
Since both $\mathcal{I}$ and $\mathcal{V}$ are partitions of $[0,1]$ it is clear that for every $a_k$ there is some $V_i$ such that $a_k\in V_i$ and, therefore, $|\tilde{J}|\leq K$. 
Note that the elements of the partition $\mathcal{V}$ which intersects more than one element of the partition $\mathcal{I}$ corresponds to sets where the function is not $\theta$\nobreakdash-Lipschitz or has discontinuity points.
Now, since the DiKernel $W$ is a bounded and measurable function, there is a number $M>0$ such that $|W(x,y)|\leq M$. This number allows us to bound the distance between $W$ and $W_\mathcal{V}$ in these discontinuity points as follows:

\begingroup
\allowdisplaybreaks
\begin{align*}
    \Delta_{\tilde{J}}&=\sum_{i,j\in \tilde{J}}\ \int_{V_i\times V_j}\left|W(x,y)-W_\mathcal{V}(x,y)\right|\mathrm{d}x \mathrm{d}y \leq\sum_{i,j\in \tilde{J}}\ \int_{V_i\times V_j}M\mathrm{d}x \mathrm{d}y\\
    & = M \sum_{i,j\in \tilde{J}}\ \lambda(V_i)\lambda(V_j) \leq M\frac{|\tilde{J}|^2}{n^2}  \leq M\frac{K^2}{n^2}
\end{align*}
\endgroup

Then $\|W-W_\mathcal{V}\|_\square\leq  \Delta_{\tilde{I}} + \Delta_{\tilde{J}} \leq \frac{2\theta}{n}+M\frac{K^2}{n^2}$, which is the desired result.
\end{proof}

\begin{proof}[Proof of Theorem~\ref{thm:dynamic_convergence}]

Fix the integer $n_0$ such that $n_0>\frac{8\theta+\sqrt{64\theta^2+16K^2M\eta}}{2\eta}$ and consider the DiKernel $W_{\mathcal{V}}$ from partition $\mathcal{V}=\{V_i\}_{i\in [I]}$ of $[0,1]$ such that $\forall i\in [I], \lambda(V_i)\leq\frac{1}{n}$ with $n>n_0$. 
Then using the result from Lemma~\ref{lem:distance_partition2} we know that $\| W-W_{\mathcal{V}}\|_{\square}\leq \frac{2\theta}{n}+\frac{MK^2}{n^2}$.
Then it follows that, ${\left\|\left(\operatorname{T}(W)f\right)(x)-\left(\operatorname{T}(W_{\mathcal{V}})f\right)(x)\right\|}_1\leq 4 \| W-W_{\mathcal{V}}\|_{\square} \leq \frac{8\theta}{n}+\frac{4MK^2}{n^2} \leq \frac{8\theta}{n_0}+\frac{4MK^2}{n_0^2} <\eta$.

Then the result follows from  Proposition \ref{prop:dynamic_partition} and Proposition\ref{prop:discounted_partition}.

\end{proof}

\section{Supplementary Material for Section~\ref{sec:competition}}\label{appendix:competition}

\begin{proposition}\label{prop:competition_utility} If the operator $C:L^1([0,1])\times S_1\times S_2\to L^1([0,1])$ satisfies Assumption~\ref{assump:competition_operator} then for every $i\in I$, the utility function $U_i$ is concave in $s_i$ and continuous in $s_i$ and $s_{-i}$. 
\end{proposition}
\begin{proof}
Note that the composition of the operator $C$ with operator $\operatorname{T}(W)$ preserves the concavity in $s_1$ and convexity in $s_2$ since the operator $\operatorname{T}(W)$ is linear.
Now for Player~1 if we compose $\operatorname{T}(W)\circ C$ with $u_1$ the concavity in $s_1$ is preserved since $u_1$ is an integral operator, on the other side for Player~2 if we compose $\operatorname{T}(W)\circ C$ with $u_2$ the resulting operator is concave in $s_2$ thanks to the negative sign in $u_2$. 
Then we can conclude that $U_i$ is concave in $s_i$ since $U_i$ is a linear operator which is the discounted sum of $u_i\circ \operatorname{T}(W)\circ C$. The continuity of $U_i$ follows an analogous argument. 
\end{proof}

\begin{proof}[Proof of Theorem~\ref{thm:existence}]
We aim at applying a generalized version proposed by \cite{glicksberg1952further} of the Nash Theorem established in \cite{nash1950equilibrium}. 
First, from Assumption~\ref{assump:strategy_set} we know that the strategy set of the lobbies is compact and using the Cartesian product it must follow that $S_1\times S_2$ is also nonempty, convex and compact. 
Second, from Assumption~\ref{assump:competition_operator} we know that the operator $C:L^1([0,1])\times S_1\times S_2\to L^1([0,1])$ is concave in $s_1$ and convex in $s_2$ then, by Proposition~\ref{prop:competition_utility}, we know that $U_i$ is concave in $s_i$. Since we have that for each $i\in I$ the space of strategies $S_i$ is compact and convex and the utility function $U_i(s_i, s_{-i})$ is continuous in $s_i$ and $s_{-i}$, and concave in $s_i$ we can apply the generalization of Nash Theorem established by \cite{glicksberg1952further}. 

For completeness we present the proof applying Kakutani's fixed\nobreakdash-point theorem to the best\nobreakdash-response correspondence.
We need to verify if the best-response correspondence satisfies the necessary conditions to apply the fixed-point theorem. 
Let us define the best\nobreakdash-response correspondences $B_i:S_{-i}\to 2^{S_i}$ defined by $B_i(s_{-i})=\argmax_{s_i\in S_i} U_i(f_0,s_i,s_{-i})$. 
Note that for every $s_{-i}\in S_{-i}$ the set $B(s_{-i})$ is nonempty since $S_{i}$ is compact and the function $U_i$ is continuous in $\tau$ then $U_i(f,s_i,s_{-i})$ is also compact, so by the Extreme Value Theorem there is $s_i^*\in S_{i}$ such that $U_i(f,s_i^*,s_{-i})\geq U_i(f,s_i,s_{-i})$ for every $s_i\in S_{i}$ and therefore $s_i^*\in B_i(s_{-i})$.
To show that $B_i(s_{-i})$ is closed we must construct a convergent sequence $(p_k)_{k\in \mathbb{N}}$ in $B_i(s_{-i})$ with limit $p$ and show that $p\in B_i(s_{-i})$. 
To do so, note that $p_k\in B_i(s_{-i})$ implies that $U_i(f,p_k,s_{-i})\geq U_i(f,s_i,s_{-i})$ for every $s_i\in S_i$, then by the continuity of $U_i$ we know that $U_i(f,p,s_{-i})\geq U_i(f,s_i,s_{-i})$ and therefore $p\in B_i(s_{-i})$. 
Finally, the convexity follows from the concavity of $U_i$. 

Let us define $B:S_1\times S_2 \to 2^{S_1\times S_2}$ defined by $B(s_1,s_2)=\prod_{i\in I} B_i(s_{-i})$ which is a nonempty, closed and convex set since the Cartesian product preserves these properties.
Then we can apply the fixed-point theorem to conclude that there is $(s_1,s_2)^*\in S_1\times S_2$ such that $(s_1,s_2)^*\in B((s_1,s_2)^*)$, which concludes the proof of existence of the equilibrium of the game.
\end{proof}

\begin{proof}[Proof of Example~\ref{ex:weighted_average_competition}]
We start by the Arzel\`a\nobreakdash-Ascoli theorem, which shows that $\mathcal{S}$ is compact for the $\|.\|_\infty$-norm.
In addition, this set is also convex since the point-wise convex combination of two bounded $\theta\text{\nobreakdash-Lipschitz}$ functions is also bounded and $\theta\text{\nobreakdash-Lipschitz}$.

Let us now focus on
\[
C(f_0,s_1,s_2)(x)=\frac{s_1(x)-s_2(x) + f_0(x)s_0(x)}{s_1(x)+s_2(x)+s_0(x)}.
\]
We prove first that the operator $C$ is concave in $s_1$ and convex in $s_2$. 
Consider the operator $\tilde{C}:S_1\to L^1([0,1])$ as $\tilde{C}(s_1)=\frac{s_1(x)+\beta}{s_1(x)+\alpha}$ as the simplified form of the operator when the lobby 1 considers every else as a constant and $\beta=-s_2(x)+f_0(x)s_0(x)$ and $\alpha=s_2(x)+s_0(x)$. 
To evaluate the concavity consider $s_1, s_1^\prime\in S_1$ and $\rho\in[0,1]$ and note that $\tilde{C}(\rho s_1+(1-\rho)s_1^\prime)=\frac{(\rho s_1+(1-\rho)s_1^\prime)(x)+\beta}{(\rho s_1+(1-\rho)s_1^\prime)(x)+\alpha} = \frac{\rho s_1(x)+(1-\rho)s_1^\prime(x)+\beta}{\rho s_1(x)+(1-\rho)s_1^\prime(x)+\alpha}$. 
Since the functions $s_1$ and $s_1^\prime$ are real valued-functions, then for every $x\in[0,1]$ there are real numbers $x_1$ and $x_1^\prime$ such that $s_1(x)=x$ and $s_1^\prime(x)=x_1^\prime$.
Then we can rewrite $\tilde{C}(\rho s_1+(1-\rho)s_1^\prime)=\frac{\rho x_1+(1-\rho)x_1^\prime+\beta}{\rho x_1+(1-\rho)x_1^\prime+\alpha}$ which is a concave real-valued function when $\alpha \geq \beta$, which is always the case, since $f_0(x)\leq 1$ for every $x$. 
In addition, $\tilde{C}$ is $\frac{\alpha-\beta}{\alpha^2}\text{\nobreakdash-Lipschitz}$ since $|\tilde{C}(s_1)-\tilde{C}(s_1^\prime)|\leq \left|\frac{s_1+\beta}{s_1+\alpha}-\frac{s_1^\prime+\beta}{s_1^\prime+\alpha}\right|\leq \left|\frac{(\alpha-\beta)s_1-s_1^\prime}{(s_1+\alpha)(s_1^\prime+\alpha)}\right|\leq \frac{(\alpha-\beta)}{\alpha^2}\left|(s_1-s_1^\prime)\right|$. 
Then we can conclude that the operator $C$ is concave in $s_1$. 
Analogously we can prove that the operator $C$ is convex in $s_2$. 
Furthermore, the operator $C$ is Lipschitz in $s_1$ and $s_2$.

Hence,the set $\mathcal{S}$ and the operator $C$ satisfy Assumptions~\ref{assump:strategy_set} and~\ref{assump:competition_operator} then we know by Theorem~\ref{thm:existence} that the equilibrium of the game exists. 
\end{proof}

\begin{proof}[Proof of Theorem~\ref{thm:epsilon_nash}]
    The proof of this result follow from Theorems~\ref{thm:dynamic_convergence} and~\ref{thm:existence}.
\end{proof}

\end{appendix}
\end{document}